\documentclass[12pt]{article}

\usepackage{amsmath,amssymb,amsfonts,amsthm,amscd,bm}
\usepackage{graphicx}
\usepackage{enumerate}
\usepackage{tikz,pgfplots}
\pgfplotsset{compat=1.18}
\usepackage{url}
\usepackage{proba}
\usepackage{hyperref}
\usepackage[utf8]{inputenc}
\usepackage[english]{babel}
\usepackage{subcaption}

\usepackage{setspace}
\def\spacingset#1{\renewcommand{\baselinestretch}{#1}\small\normalsize}
\spacingset{1}

\usepackage{booktabs,caption}
\usepackage[flushleft]{threeparttable}

\makeatletter
\g@addto@macro\normalsize{%
  \setlength\abovedisplayskip{6pt}
  \setlength\belowdisplayskip{6pt}
  \setlength\abovedisplayshortskip{4pt}
  \setlength\belowdisplayshortskip{4pt}}
\makeatother

\usepackage{titlesec}
\titlespacing{\section}{0pt}{*0.2}{*0.2}
\titlespacing{\subsection}{0pt}{*0.2}{*0.2}
\titlespacing{\subsubsection}{0pt}{*0.2}{*0.2}

\usepackage{etoolbox}
\makeatletter
\patchcmd{\thm@space@setup}
  {\thm@preskip = \topsep}
  {\thm@preskip = 2pt}
  {}{}
\patchcmd{\thm@space@setup}
  {\thm@postskip = 0pt}
  {\thm@postskip = 2pt}
  {}{}
\makeatother

\usepackage[ruled,vlined]{algorithm2e}
\SetKwInput{KwInput}{Input}
\SetKwInput{KwOutput}{Output}

\DeclareFontFamily{OT1}{pzc}{}
\DeclareFontShape{OT1}{pzc}{m}{it}{<-> s * [1.10] pzcmi7t}{}
\DeclareMathAlphabet{\mathpzc}{OT1}{pzc}{m}{it}

\usepackage[natbibapa]{apacite}

\hypersetup{
  pdfstartview={XYZ null null 1.50},
  colorlinks=true,
  citecolor=blue,
  urlcolor=blue,
  linkcolor=blue
}

\usepackage{multirow,placeins,capt-of}
\tikzset{every picture/.style={line width=0.75pt}}
\newtheorem{lemma}{Lemma}

\newtheorem{cor}{Corollary}
\newtheorem{theorem}{Theorem}

\newtheorem*{remark}{Remark}
\newtheorem{definition}{Definition}
\newtheorem{assumption}{Assumption}

\begin{document}

% Title and authors
%Potential Outcome Modeling and Estimation\\ in DiD Designs with Staggered Treatments
{\title{\bf 
Bayesian Estimation \\of Cohort-Time-Stratum Specific Effects \\
in Staggered Difference-in-Differences\thanks{Email: \url{chib@wustl.edu} and \url{kenichi.shimizu@ualberta.ca}}
}
\author{
\begin{tabular}[t]{c@{\extracolsep{8em}}c}
Siddhartha Chib & Kenichi Shimizu \\
Olin School of Business & Department of Economics \\
Washington University & University of Alberta, \\
in St.\ Louis, USA & Canada \\
\end{tabular}}
\date{May 2026}
\maketitle
}
\vspace{-.50in}
\begin{abstract} %100 WORD LIMIT (JAE)--- Currently it is at 99.
Difference-in-Differences designs with staggered treatment adoption are widely used to study heterogeneous treatment effects across cohorts and time periods. We develop a probabilistic framework for estimating potentially high-dimensional ATT arrays that vary across cohorts, periods, and strata defined by baseline covariates. The framework jointly estimates subgroup-specific treatment effects through a unified likelihood-based model, stabilizing inference in sparse cohort-by-time-by-stratum settings. We establish a Bernstein–von Mises theorem for the ATT array, implying asymptotically valid frequentist coverage of posterior credible intervals. Simulations and an application to minimum wage increases and teen employment demonstrate meaningful finite-sample improvements and important subgroup heterogeneity.

\end{abstract}


\newpage
%\spacingset{1.8} % DON'T change the spacing!
\onehalfspacing

\section{Introduction}

Difference-in-Differences (DiD) methods with staggered treatment adoption are widely used to evaluate policies and interventions. A central insight of the recent literature is that conventional two-way fixed effects (TWFE) estimators generally fail to recover causally interpretable  effects when treatment timing varies across units and treatment effects are heterogeneous (\citealp*{deChaisemartin_dHaultfoeuille2020AER, Goodman2021JoE,SunAbraham2021dynamic,AtheyImbens2022design,BorusyakJaravelSpiess2024Restud}). 
Modern approaches instead focus on identifying cohort- and time-specific Average Treatment Effects on the Treated (ATTs) under parallel trends and no-anticipation assumptions (\citealp{CallawaySantAnna2021did}).

Existing staggered DiD procedures transparently identify cohort-time treatment effects under the identifying assumptions and provide aggregation schemes for summarizing them. However, disaggregated subgroup-specific ATTs are often not the primary target of these procedures. On the other hand, in many empirical applications, researchers  seek to understand how treatment effects vary across subpopulations defined by baseline characteristics. For example, the effects of a labor market intervention may differ systematically across regions, demographic groups, or firm sizes. \cite*{Roth2023JoE} (Section 6) highlight this growing interest in conditional treatment effect heterogeneity within modern DiD designs. However, once ATT parameters are indexed jointly by treatment cohort, time period, and covariate-defined stratum, estimation becomes a high-dimensional problem with many sparse cohort-by-time-by-stratum cells. Existing staggered DiD procedures can in principle be applied separately within strata but such approaches may become unstable when effective sample size per stratum is small.

This paper develops a Bayesian framework for estimating heterogeneous ATTs in staggered DiD designs. We study ATT parameters indexed jointly by treatment cohort, time period, and strata defined by pre-treatment covariates, allowing treatment effects to vary flexibly across subpopulations. Our approach discretizes the baseline covariate space into user-defined strata and jointly models untreated and treated potential outcomes, imposing the identifying assumptions of modern staggered DiD methods. The resulting likelihood-based representation treats the potentially high-dimensional ATT array as an explicit parameter object, facilitating information sharing across related cells and regularization. 

The framework is particularly useful in settings where some cohort-by-time-by-stratum cells contain relatively few observations. Our framework jointly estimates the ATTs across cohorts and strata, stabilizing finite-sample estimation relative to approaches that estimate subgroup-specific ATTs separately. Simulation studies show finite-sample advantages of our joint-estimation approach over an existing method  applied separately on each stratum. 

The proposed framework also provides a probabilistic approach to assessing pre-treatment restrictions. Rather than conditioning on conventional pre-trend tests, which are known to have low power and may induce post-selection distortions (\citealp*{FreyaldenhovenChristianShapiro2019AER,Roth2022AERI,Roth2023JoE}), we compare models with and without pre-treatment restrictions using marginal likelihoods and propagate uncertainty about pre-trends directly into posterior inference on ATTs.

The same marginal likelihood framework can also be used to compare alternative heterogeneity specifications, although formal selection of the stratification structure is not the primary focus of this paper. This provides a probabilistic measure of how much subgroup heterogeneity is supported by the data while accounting for the increased complexity of richer specifications. At the same time, the degree of stratification ultimately remains a substantive modeling choice that may also reflect policy or scientific considerations beyond statistical fit.

Finally, we establish a Bernstein--von Mises theorem for the ATT array, implying that posterior credible intervals possess asymptotically valid frequentist coverage under standard regularity conditions. 
An application to minimum wage increases and teen employment in the United States reveals substantially larger employment responses in smaller counties than in larger counties, illustrating the importance of subgroup-specific ATT estimation.

\textbf{Related literature:}
Our paper contributes to several related literature. First, we build on work studying treatment effect heterogeneity in DiD designs using observed pre-treatment covariates. \citet{abadie2005semiparametric} develops semiparametric reweighting methods for heterogeneous DiD effects, while \citet*{ImaiQinYanagi2025doubly} study cohort-time conditional treatment effect functions with a continuous pre-treatment covariate. In contrast, we focus on estimating treatment effects that  vary jointly across treatment cohorts, time periods, and strata defined by potentially multiple pre-treatment covariates. Second, we contribute to the growing literature on Bayesian and empirical Bayes methods for DiD and event-study designs, including \citet*{BreunigLiuYu2024semiparametric} and \citet{BotosaruLiu2025time}. Relative to these approaches, we focus on joint modeling and estimation of subgroup-specific ATTs  with sparse cohort-by-time-by-stratum cells while preserving the estimand structure of modern DiD procedures, particularly the grouped ATT framework of \citet{CallawaySantAnna2021did}. 
Third, unlike imputation-based approaches that estimate untreated outcome models using untreated observations and then recover counterfactuals for treated units (\citealp*{BorusyakJaravelSpiess2024Restud}; \citealp*{GardnerThakralToYap2025}) or pooled regression approaches such as \cite{Wooldridge2025}, 
we specify a unified probabilistic model for both treated and untreated outcomes. This formulation allows the ATT array to be expressed directly as a linear function of model parameters, facilitating joint estimation and regularization across cohort-time-stratum cells. These features are particularly important when treatment effects are defined for fine subgroups, where effective sample sizes within cells can be small and separate estimation may become unstable.

The remainder of the paper proceeds as follows. 
Section \ref{sec:model} presents the model.  
Section \ref{sec:estimation} describes estimation and inference. 
Section \ref{sec:simulation_application} reports simulations and an empirical application. 
Online Appendix contains additional results on technical details, simulations, and application.

   %Roth etal (Section 6) :Conditional treatment effects. The DiD literature has placed a lot of emphasis on learning about the ATT’s of different groups. However, in many situations, it may also be desirableto better understand how these ATT’s vary between subpopulations defined by covariate values. For instance, how does the average treatment effect of a training program on earnings vary according to the age of its participants?

\section{The model}\label{sec:model}

We consider an event-study setting in which units are treated at different times.
We observe a random sample of $n$ units over $T$ time periods. 
Units $i=1,\dots,n$ are observed over periods $t=1,\dots,T$.
Let $y_{it}$ be the observed outcome for unit $i$ at time $t$. 
Let $\mathcal{S}\subset \{2,\ldots,T\}$ be the set of treatment timings. 
We let $S_i$ be the treatment cohort. 
For  treated units, $S_i\in\mathcal{S}$ indicates the first treated period, and as a convention, we let $S_i=1$ for never-treated units. 
We suppose that never-treated units exist. 
Each unit  is randomly assigned (ex-ante researcher's perspective) to one of the cohorts in $\{1\}\cup \mathcal{S}$. 

Let $D_{it}$ be the treatment indicator. Treatment is absorbing in the sense that 
$D_{it}=\mathbf{1}\{t\ge S_i\}$ if $S_i\ge 2$. 
$D_{it}=0 \ \forall t$ if $S_i=1$. 

We denote baseline covariates by $w_i \in \mathcal{W}\subset \mathbb{R}^{K}$, where $\mathcal{W}$ is compact. The vector $w_i$ does not include a constant term
for identification purposes. We do not use post-treatment covariates in order to avoid
potential endogeneity.

\subsection{Potential outcomes and treatment effects}
Let $Y_{it}(0)$ be unit $i$'s potential outcome if it was never-treated (i.e. $S_i=1$). 
For $s\in\mathcal{S}$, let $Y_{it}(s)$ be the potential outcome of unit $i$ if it gets treatment in period $s$ (i.e. $S_i=s$). The observed and potential outcomes are related through 
$y_{it}=Y_{it}(0)\cdot 1(S_i=1)+\sum_{s\in\mathcal{S}}Y_{it}(s)\cdot 1(S_i=s)$.

For $s\in \mathcal{S}$,  the average treatment effect for treated (ATT) for units treated at period $s$ measured at time $t$, conditional on pre-treatment covariate $w$ (with an abuse of notation on conditioning on possibly continuous covariates), is defined as 
\begin{equation}
    \tau_{\mathrm{ATT}}(s,t; w) = \E[Y_{it}(s)-Y_{it}(0) \vert S_i = s, w_i=w] \label{eq:ATT}.
\end{equation}

\subsubsection{Motivation for stratification and  potential outcome modeling}
As the ATT \eqref{eq:ATT} is cohort-time-covariate specific, one can quickly end up with numerous parameters. \cite{CallawaySantAnna2021did} address this by aggregating the ATTs at least across the covariate space, and possibly also over cohort and/or time.  

However, in some applications,  researchers may wish to retain the level of heterogeneity for policy purposes, as described in introduction. We exploit two key ideas to enable  estimation of the potentially high-dimensional ATTs that depend on cohort, time, and covariates. The first is  discretization of the covariate space and the second is laying out a tractable model that  relates the treatment effect parameters to the evolution of the potential outcomes while satisfying the identifying assumptions. 
Below, we first describe the stratification and, in the next section, introduce our model for potential outcomes. 

We assume that the units are  divided into $G$ known strata or clusters defined by the baseline covariate $w_i$ observed before treatments. For example, when $w_i$ is one-dimensional, $G-1$ thresholds on the real line define the strata. When $w_i$ is multi-dimensional, the whole vector or a subvector of $w_i$ can be used to define the strata.  

In this paper, we take the strata as known. The most natural partitioning rule is to divide the covariate space so that the number of units is roughly equal across strata in each cohort. 
In this paper, we do not tackle the problem of searching over  latent partitions. 
Let $G_i\in \{1,\ldots,G\}\equiv \mathcal{G}$ be the stratum that unit $i$ belongs to. 
The ATTs \eqref{eq:ATT} are now indexed by  cohort, time, and stratum. 
\begin{definition}[Cohort-time-stratum specific ATTs]
The ATT for units 
treated at period $s\in \mathcal{S}$, 
belonging to stratum $g\in\mathcal{G}$, and 
measured at time $t$ is defined as 
\begin{equation}
    \tau_{\mathrm{ATT}}(s,t; g) = \E[Y_{it}(s)-Y_{it}(0) \vert S_i = s, G_i=g] \label{eq:ATT_stratum}.
\end{equation}    
\end{definition}

The cohort-time-stratum  ATTs are  identified under the following two  conditions assumed to hold within strata. Across the entire population, they need not hold if cohort assignment is correlated with the covariates. We omit the identification proof as it is a straightforward application of the standard argument (e.g.\ \cite{CallawaySantAnna2021did}) within strata.
\begin{assumption}[No-anticipation (NA)]\label{assm:NA}
For each stratum $g\in\mathcal{G}$, 
expected values of the treated and untreated potential outcomes for cohort
$s \in \mathcal{S}$ coincide in the pre-treatment periods:
\[
\E[Y_{it}(s)\vert S_i = s, G_i=g]=\E[Y_{it}(0)\vert S_i = s, G_i=g],
\; t<s .
\]
\end{assumption}
This assumption rules out anticipation effects for units that are eventually treated
and is likely to hold when treatment timing is not chosen by the units themselves.
Under this assumption, $\tau_{\mathrm{ATT}}(s,t;g)=0$ for all $t<s$ and $g$.

\begin{assumption}[Parallel Trends (PT)]\label{assm:PT}
For each stratum $g\in\mathcal{G}$, for each treatment cohort $s\in\mathcal{S}$, the expected increment in the untreated
potential outcome coincides with that of the never-treated units in all
post-treatment periods:
\[
\E\!\left[Y_{it}(0)-Y_{i,t-1}(0)\vert S_i = s, G_i=g\right]
=
\E\!\left[Y_{it}(0)-Y_{i,t-1}(0)\vert S_i = 1, G_i=g\right],
\;   t\ge s .
\]
\end{assumption}
This assumption states that, after treatment adoption, the counterfactual
(untreated) outcome increments for treated cohorts evolve in parallel with those
of the never-treated units. That is, absent treatment, treated and never-treated
units would have experienced the same expected changes in outcomes in 
post-treatment periods.

When deciding the strata, there are two important considerations to be made. 
First, the strata should be defined so that the ATTs are useful for policy recommendations. 
For example, in medical applications,  it would be informative to learn effects of a new medicine specific to some demographic variable, such as age (e.g. young vs old), but probably not other variables such as income. 
Second, the identifying assumptions need to be realistic within  strata. 

\textit{Remark:} 
While the proposed framework stabilizes estimation through joint modeling across cohort-time-stratum cells, it is not exempt from dimensionality considerations. In particular, excessively fine stratification may leave insufficient information within cells to reliably estimate common model components and calibrate prior hyperparameters. Thus, the choice of strata should balance substantive policy relevance against statistical precision. Nevertheless, the framework remains applicable in settings with relatively small numbers of observations per cell, where conventional asymptotic approximations underlying separate stratum-by-stratum estimation may become unreliable. Later, we discuss how marginal likelihood can be used as a diagnostic tool that can warn against an extremely fine stratification that data cannot support.  

%The choice of the number of strata $G$ reflects a tradeoff between modeling richer treatment effect heterogeneity and maintaining sufficient statistical precision within cohort-by-time-by-stratum cells. While finer stratifications may uncover substantively important heterogeneity, excessively large values of $G$ can lead to weakly identified models with limited information per cell. In the proposed framework, marginal likelihood comparison provides a probabilistic criterion for assessing the extent of heterogeneity supported by the data. At the same time, model comparison need not mechanically determine the preferred specification: researchers may still choose richer stratifications motivated by substantive or policy considerations even when simpler specifications receive greater support from the marginal likelihood.

\subsection{The model for potential outcomes}
We now describe our model for potential outcomes. The goal is to lay out a tractable model that relates the ATT parameters to the  sequence of potential outcomes while explicitly imposing
the parallel trends and no-anticipation assumptions (Assumptions \ref{assm:NA} and \ref{assm:PT}). 
We begin with the never-treated units. 
Consider $\E[Y_{it}(0)\vert S_i=1, G_i=g]$, the expected outcome for the never-treated units in stratum $g\in \mathcal{G}$. Without loss of generality, we let 
$\beta_{11,g}=\E[Y_{i1}(0)\vert S_i=1, G_i=g]$ denote the mean initial value for the never-treated units in this stratum. 
We then let $b_{12,g}=\E[Y_{i2}(0)\vert S_i=1, G_i=g]-\E[Y_{i1}(0)\vert S_i=1, G_i=g]$ be the increment from period $1$ to $2$. Generally, by letting $b_{1k,g}$ be the increment from period $k-1$ to $k$, we can write 
\begin{align}
    \E[Y_{it}(0)\vert S_i=1, G_i=g]&=\beta_{11,g} + \sum_{k=2}^{t}b_{1k,g}, \ t=1,\ldots,T .     \label{eq:PO_0_seq_1}
\end{align} 
Similarly, for each treated cohort $s\in\mathcal{S}$ and stratum $g\in\mathcal{G}$, we let, for $t=1,\ldots,T$, 
\begin{equation}\label{eq:PO_1_seq_s}
\E[Y_{it}(s)\vert S_i=s, G_i=g]=
\beta_{s1,g} + \sum_{k=2}^{t}b_{sk,g} =
\begin{cases}
\beta_{s1,g}+\sum_{k=2}^{t} b_{sk,g}, 
& \text{if } t<s, \\[2pt]
\beta_{s1,g}+\sum_{k=2}^{s-1} b_{sk,g}+\sum_{k=s}^{t} b_{sk,g}, 
& \text{if } t\ge s.
\end{cases}
\end{equation}
The initial value $\beta_{s1,g}$ and increment parameters  $b_{sk,g}$ are all specific to the cohort $s$ and stratum $g$. The parameters in the increment representations for the observed outcomes \eqref{eq:PO_0_seq_1} and \eqref{eq:PO_1_seq_s} are unique. We will revisit this point shortly. 

The unique feature of our approach is to now model  the counterfactual untreated potential outcomes based on \eqref{eq:PO_0_seq_1} and \eqref{eq:PO_1_seq_s}, imposing the identifying assumptions: for $t=1,\ldots,T$, we let 
\begin{equation}\label{eq:PO_0_seq_s}
\E[Y_{it}(0)\vert S_i=s, G_i=g]=
\begin{cases}
\beta_{s1,g}+\sum_{k=2}^{t} b_{sk,g}, 
& \text{if } t<s, \\[2pt]
\beta_{s1,g}+\sum_{k=2}^{s-1} b_{sk,g}+\sum_{k=s}^{t} b_{1k,g}, 
& \text{if } t\ge s.
\end{cases}
\end{equation}
By construction, this model satisfies NA and PT. NA is satisfied because 
$\E[Y_{it}(s)\vert S_i=s, G_i=g]$ and $\E[Y_{it}(0)\vert S_i=s, G_i=g]$ coincide for all
$t<s$. PT is satisfied since the increments in the untreated
outcomes in the treated cohorts,
$\E[Y_{it}(0)-Y_{i,t-1}(0)\vert S_i=s, G_i=g]$, for $t\ge s$, match those of the
never-treated units,
$\E[Y_{it}(0)-Y_{i,t-1}(0)\vert S_i=1, G_i=g]$. 
In this way, the model explicitly incorporates the restriction that, absent treatment, post-treatment outcome dynamics evolve identically across cohorts.

This model can be concisely presented in a vector-matrix form. Define the $T\times 1$ vectors  
$Y_{i}(0) =(Y_{i1}(0),\ldots,Y_{iT}(0))'$, and for $s\in\mathcal{S}$, 
$Y_{i}(s) =(Y_{i1}(s),\ldots,Y_{iT}(s))'$. 
Also, let 
\[
\beta_{sg} =
    \big(\beta_{s1,g},
    b_{s2,g},
    \cdots, 
    b_{sT,g}\big)',
\quad 
s\in \{1\} \cup \mathcal{S} \text{ and } g\in\mathcal{G}.
\]
Finally, let $L_T$ denote the $T\times T$  lower triangular matrix of ones, $L_s^{\text{pre}}$ the matrix obtained by setting columns $s$ to $T$ of $L_T$ to zero,
and $L_s^{\text{post}}$ the matrix obtained by setting columns 1 to $s-1$ of $L_T$ to zero. For example, when $T = 5$ and $s = 4$, these matrices are of the form
\[
L_T =
\begin{bmatrix}
1 & 0 & 0 & 0 & 0 \\
1 & 1 & 0 & 0 & 0 \\
1 & 1 & 1 & 0 & 0 \\
1 & 1 & 1 & 1 & 0 \\
1 & 1 & 1 & 1 & 1
\end{bmatrix}, \quad
L_s^{\text{pre}} =
\begin{bmatrix}
1 & 0 & 0 & 0 & 0 \\
1 & 1 & 0 & 0 & 0 \\
1 & 1 & 1 & 0 & 0 \\
1 & 1 & 1 & 0 & 0 \\
1 & 1 & 1 & 0 & 0
\end{bmatrix}, \quad
L_s^{\text{post}} =
\begin{bmatrix}
0 & 0 & 0 & 0 & 0 \\
0 & 0 & 0 & 0 & 0 \\
0 & 0 & 0 & 0 & 0 \\
0 & 0 & 0 & 1 & 0 \\
0 & 0 & 0 & 1 & 1
\end{bmatrix}.
\]
It is easily seen that 
$$
L_T = L_s^{\text{pre}}+L_s^{\text{post}},
$$
which we use repeatedly. Then, the models for the observed outcome vectors, namely  
\eqref{eq:PO_0_seq_1}
and 
\eqref{eq:PO_1_seq_s}, 
can be written as 
\begin{align}
    \mathbb{E}[Y_{i}(0)\vert S_i=1, G_i=g]&=L_T \beta_{1g} = L_s^{\text{pre}} \beta_{1g} +L_s^{\text{post}} \beta_{1g}, \label{eq:PO_0_seq_1_mat}\\
    \mathbb{E}[Y_{i}(s) \vert S_i=s, G_i=g]&=L_T \beta_{sg} = L_s^{\text{pre}} \beta_{sg} +L_s^{\text{post}} \beta_{sg}.\label{eq:PO_1_seq_s_mat}  
\end{align}
As $L_T$ is invertible, $\beta_{1g} $ and $\beta_{sg}$ are unique in these representations. 
Similarly, we write \eqref{eq:PO_0_seq_s} as
\begin{equation}
    \mathbb{E}[Y_{i}(0) \vert S_i=s, G_i=g]=L_s^{\text{pre}} \beta_{sg} + L_s^{\text{post}} \beta_{1g}.\label{eq:PO_0_seq_s_mat}      
\end{equation}
$L_s^{\text{pre}}$ and $L_s^{\text{post}}$ capture the increments for the pre-treatment periods and post-treatment periods, respectively. 
It now follows that the $T$-dimensional ATT vector for cohort $s$ and stratum $g$ is given by
    \begin{align}
        \tau_{ATT}^{(sg)} &= L_s^{\text{post}} (\beta_{sg} - \beta_{1g}), \label{eq:att_matrix_form}
    \end{align}
    where the typical element is
    $
    \tau_{ATT}(s,t;g) = \mathbb{E}[Y_{it}(s) - Y_{it}(0)\vert S_i=s, G_i=g] = \sum_{k=s}^t (b_{sk,g} - b_{1k,g})  \text{ for } t = s, \ldots, T,
    $
    and $\tau_{ATT}(s,t;g) =0$ for $t<s$ by the no-anticipation assumption.
Thus, the ATTs are functions of the parameters and take the form of difference-in-differences.

\subsection{Re-parametrization}
The models 
\eqref{eq:PO_0_seq_1_mat}, 
\eqref{eq:PO_1_seq_s_mat}, and 
\eqref{eq:PO_0_seq_s_mat}  
allow for arbitrary evolution of the potential outcomes under the identifying restrictions, but are over-parameterized. This is evident from the expression for the ATT in \eqref{eq:att_matrix_form}, which depends on the full set of increment parameters. We achieve a more parsimonious formulation by letting 
\[
\beta_{sg} = \beta_{1g} + \delta_{sg}, \quad s \in \mathcal{S} \text{ and } g \in \mathcal{G},
\]
where $\delta_{sg}$ represents the difference in increments between cohort $s$ and the never-treated units, within stratum $g$. 
Substituting $\beta_{sg}$ in 
\eqref{eq:PO_1_seq_s_mat}
and 
\eqref{eq:PO_0_seq_s_mat},
we get 
\begin{align}
    \mathbb{E}[Y_{i}(s) \vert S_i=s, G_i=g]
    &=L_T \beta_{1g} + L_T \delta_{sg}.  \label{eq:PO_1_seq_s_repara_mat}  \\
    \mathbb{E}[Y_{i}(0) \vert S_i=s, G_i=g]
    &=L_s^{\text{pre}} (\beta_{1g} + \delta_{sg}) + L_s^{\text{post}} \beta_{1g} 
    =L_T \beta_{1g} + L_s^{\text{pre}} \delta_{sg},      \label{eq:PO_0_seq_s_repara_mat}    
\end{align}

\begin{remark}
Subtracting (\ref{eq:PO_0_seq_s_repara_mat}) from (\ref{eq:PO_1_seq_s_repara_mat}) we get
\begin{align}
    \tau_{ATT}^{(sg)} &= L_s^{\text{post}} \delta_{sg},
    \label{eq:att_repara_matrix}
\end{align}
where the typical element of this vector is:
\begin{align}
    \tau_{ATT}(s,t;g) &= \mathbb{E}[Y_{it}(s) - Y_{it}(0) \vert S_i=s, G_i=g] = \sum_{k=s}^t \delta_{sk,g}, 
    \quad \text{for } t = s, \ldots, T,
    \label{eq:att_repara_element}
\end{align}
and $\tau_{ATT}(s,t;g) =0$ for $t<s$ by the no-anticipation assumption. 
In this parameterization, the number of parameters involved in 
the definition of the ATT, relative to \eqref{eq:att_matrix_form}, is reduced by half. 
\end{remark}

The  reparametrization simplifies the selection of priors. 
For instance, if one believes ex-ante  that the ATTs are zero for every post-period $t \geq s$, this can be encoded by centering the prior distribution of $\delta_{st,g}$ around zero. Moreover, knowledge about the ATTs from related studies can also be utilized. 
%Note that $\delta_{ss,g}=\tau_{ATT}(s,s;g)$ and $\delta_{st,g}=\tau_{ATT}(s,t;g)-\tau_{ATT}(s,t-1;g)$, $t=s+1,\ldots,T$

One may wonder whether it is possible to eliminate the cumulative summation in the definition of the ATTs by considering a simpler model like
\begin{align*}
    \mathbb{E}[Y_{it}(0)\vert S_i=1,G_i=g]&=\tilde{\beta}_{11,g}+\tilde{\beta}_{1t,g}, \ t=1,\ldots,T,\\
    \mathbb{E}[Y_{it}(s)\vert S_i=s,G_i=g]&=\tilde{\beta}_{s1,g}+\tilde{\beta}_{st,g}, \ t=1,\ldots,T,\\
    \E[Y_{it}(0)\vert S_i=s, G_i=g]&=
\begin{cases}
\tilde{\beta}_{s1,g}+\tilde{\beta}_{st,g}, 
& \text{if } t<s, \\[2pt]
\tilde{\beta}_{s1,g}+\tilde{\beta}_{1t,g}, 
& \text{if } t\ge s.
\end{cases}
\end{align*}
so that we have the single-parameter representation of the ATTs after a reparametrization: 
$\tau_{ATT}(s,t;g) = \tilde{\beta}_{st,g}-\tilde{\beta}_{1t,g}=\tilde{\delta}_{st,g}$. 
However, this formulation violates the parallel trend condition at $t=s$: 
$\E[Y_{is}(0)-Y_{i,s-1}(0)\vert S_i=s, G_i=g]=\tilde{\beta}_{1s,g}-\tilde{\beta}_{s,s-1,g}$ 
which is generally not equal to  
$\E[Y_{is}(0)-Y_{i,s-1}(0)\vert S_i=1, G_i=g]=\tilde{\beta}_{1s,g}-\tilde{\beta}_{1,s-1,g}$. 
This observation highlights the fact that the parallel trends condition constrains the increments of the expected potential outcomes rather than their levels. Hence, the cumulative-sum representation in \eqref{eq:att_repara_element} arises naturally from imposing parallel trends on the evolution of untreated potential outcomes and is not merely a modeling convenience.

\subsection{Pre-test as a model comparison}\label{sec:prePT_ML}

We now develop a Bayesian framework to evaluate whether the parallel trends (PT) assumption holds in the pre-treatment periods. While support for parallel trends in the pre-treatment periods is often taken as an evidence that the assumption is credible post-treatment \citep{CallawaySantAnna2021did}, such assessments can induce frequentist pre-test bias when inference is conditional  on test outcomes. 
Also, the conventional implementation of the pre-test does not provide clear 
guidance when the hypothesis is rejected (\citealp{Roth2023JoE}). Even in such a case, the researcher may still want to
estimate the treatment effect of interest.

We instead frame the question as one of model comparison. We estimate models that impose and relax the pre-treatment PT restriction and compare them using marginal likelihoods computed by the method of \citet{chib1995}. This strategy measures the empirical support for the identifying PT assumption directly, without conditioning inference on a preliminary test.
Recall that the parallel trends condition (Assumption~\ref{assm:PT}) is given by
\[
\E[ Y_{it}(0)-Y_{i,t-1}(0)\vert S_i=s,G_i=g ] = \E[Y_{it}(0)-Y_{i,t-1}(0)\vert S_i=1,G_i=g  ]=b_{1t,g}, \; t\geq s.
\]
Suppose now this condition also holds for the pre-treatment periods (i.e.\ for $t<s$). Together with the no-anticipation condition (Assumption \ref{assm:NA}), 
%$\E[Y^{(1)}_{s,it}]=\E[Y^{(0)}_{s,it}]$ for $t<s$, 
this leads to  
\[
\E[ Y_{it}(s)-Y_{i,t-1}(s)\vert S_i=s,G_i=g ]=b_{1t,g}, \text{ for } t<s, 
\]
which implies that $\delta_{st,g}=0$ for $t<s$. The resulting model therefore has a smaller effective number of parameters in the vector $\delta_{sg}$. 

Under the additional restriction of parallel trends on pre-treatment periods, the models   \eqref{eq:PO_1_seq_s} and \eqref{eq:PO_0_seq_s} become 
\begin{align*}
    \E[Y_{it}(s) \vert S_i=s,G_i=g ]&=
    \begin{cases}
      \beta_{s1,g} + \sum_{k=2}^{t}b_{1k,g}, & \text{if}\ t<s \\
      \beta_{s1,g} + \sum_{k=2}^{s-1}b_{1k,g}  + \sum_{k=s}^{t}b_{sk,g}, & \text{if}\ t\geq s,
    \end{cases} \\
    \E[Y_{it}(0) \vert S_i=s,G_i=g ]&=
    \begin{cases}
      \beta_{s1,g} + \sum_{k=2}^{t}b_{1k,g}, & \text{if}\ t<s \\
      \beta_{s1,g} + \sum_{k=2}^{s-1}b_{1k,g}  + \sum_{k=s}^{t}b_{1k,g}, & \text{if}\ t\geq s,
    \end{cases}     
\end{align*}
for $t=1,\ldots,T$, $s\in \mathcal{S}$, and $g\in \mathcal{G}$. 
See Online Appendix for the matrix form of this model.

It may be noted that if the researcher has a prior probability that the parallel trend condition holds (or does not hold) in the pre-treatment periods, then the corresponding posterior probabilities can be obtained based on the computed marginal likelihoods. Inference on ATTs may then be averaged across specifications using Bayesian model averaging, thereby incorporating the ex-post uncertainty regarding the pre-treatment parallel trend condition directly into  inference on ATTs. This is in contrast to the traditional approach that conditions on passing the pre-trends test and is known to have low power. 

\paragraph{Comparing different strata specifications} 
Although this is not a primary focus of the paper,  the same model comparison principle extends naturally to comparison of  alternative heterogeneity specifications. Richer stratifications introduce additional ATT parameters and greater flexibility, but they also increase model complexity and may weaken identification in sparse cohort-by-time-by-stratum cells. The marginal likelihood automatically balances these competing forces by rewarding improved fit while penalizing unnecessary complexity.

At the same time, model comparison need not mechanically determine the preferred specification. In empirical applications, researchers may prefer richer heterogeneity structures (e.g. $G=2$ over $G=1$) for substantive or policy reasons even when simpler specifications receive greater marginal likelihood support. For example, a policymaker may still wish to distinguish treatment effects across demographic or regional subgroups despite limited statistical evidence favoring additional stratification. Accordingly, we view marginal likelihood comparison primarily as a probabilistic diagnostic tool that helps quantify how much heterogeneity is supported by the data.%, rather than as an automatic model-selection device.
\section{Estimation and Inference}\label{sec:estimation}
In this section, we present an estimation approach for the baseline model without pre-period restrictions. Online Appendix  illustrates the case under the restriction that parallel trends hold in pre-periods. 
Let $y_{i}=(y_{i1},\ldots,y_{iT})'$ be the $T$-dimensional vector of observed outcomes. 
By consistency, we then have $y_{i}=Y_{i}(0)$ if $S_i=1$ and $y_{i}= Y_{i}(s)$ if $S_i=s\in \mathcal{S}$. 
We can write 
from 
\eqref{eq:PO_0_seq_1_mat}
and 
\eqref{eq:PO_1_seq_s_repara_mat}, 
\begin{align}
    y_{i}
    &=
    L_T \beta_{1g} 
    +
    \varepsilon_{i},  \hspace{2cm} \text{ if } S_i = 1 \text{ and } G_i=g,\label{eq:PO_nevertreated_matrix}\\
    y_{i}
    &=
    L_T \beta_{1g}
    +
    L_T \delta_{sg}
    +\varepsilon_{i},  \hspace{0.5cm} \text{ if } S_i =  s  \text{ and } G_i=g , \label{eq:PO_treated_obs_matrix}
\end{align}
where 
$\varepsilon_{i} =(\varepsilon_{i1},\ldots,\varepsilon_{iT})'$ is a mean-zero vector.

For simplicity, we suppose that the errors are Gaussian 
$\varepsilon_{i} \vert S_i=s \sim N_T(0, \Sigma_s), \ s\in \{1\}\cup \mathcal{S}$, 
where each of the $T \times T$ covariance matrices are assumed to be diagonal: 
$\Sigma_s =\text{diag}( \sigma^2_{s1},\ldots, \sigma^2_{sT})$. 
For greater generality, one can allow the errors to be autocorrelated and heavy-tailed. 
The proposed model can be used in conjunction with Bayesian semiparametric methods such as Dirichlet process mixtures. 
Alternatively, in order to lower the risk of mis-specification, the sandwich posterior approach proposed in 
\cite{muller2013risk} can be used in conjunction with
our model and estimation strategy. Given the potential outcome modeling for staggered
DiD laid out in this paper, such extensions are now possible.

\subsection{Correlated random effects}
In order to capture  unit-level  variation in the initial value $\beta_{s1,g}$, we introduce heterogeneous initial values $\alpha_i$. Specifically, we employ  correlated random effects:
\[
\alpha_i \vert w_i \sim N(w_i'\phi_s, D_s).
\]
The random initial values are related to $\beta_{s1,g}$ in  our formulation via: $\beta_{s1,g}=\mathbb{E}[\alpha_i\vert S_i=s, G_i=g] = \mathbb{E}[w_i\vert S_i=s,G_i=g]'\phi_s$, and the resulting model still satisfies the identifying assumptions, namely NA and PT.  

Because the heterogeneous initial values $\alpha_i$ and $\beta_{s1,g}$ are not separately identified, we impose the normalization $\beta_{s1,g}=0$, which also implies that $\delta_{s1,g}=0$. 
For this reason, we define the $(T-1)$-dimensional vectors $\eta_{1g}$ and $\xi_{sg}$ by removing the first elements in $\beta_{1g}$ and $\delta_{sg}$, respectively. While we work with these lower dimensional vectors, the original vectors can be recovered as $\beta_{1g}=R_0 \eta_{1g}= [0; \eta_{1g}']' \in \mathbb{R}^T$ and $\delta_{sg}=R_0 \xi_{sg}= [0; \xi_{sg}']' \in \mathbb{R}^T$, where  $R_0 = [0_{T-1}';I_{T-1}]$ is a $T\times (T-1)$ selection matrix.

\subsection{Likelihood}
Let $N_{sg}=\{i: S_i=s \text{ and } G_i=g \}$ be the set of units that belong to the treatment cohort $s$ and covariate stratum $g$. Let $1_T$ be the $T$-dimensional column vector of ones. 
For each stratum $g\in \mathcal{G}$, we have the following. For  $i\in N_{1g}$,  
\begin{align*}
    y_i
    &=1_T \alpha_i + L_T R_0 \eta_{1g} + \varepsilon_{i}\\
    &=1_T \alpha_i + L_T \beta_{1g} + \varepsilon_{i}, \quad \varepsilon_{i}\sim N(0,\Sigma_1)  
\end{align*}
and for $i \in N_{sg}, s\in \mathcal{S}$,
\begin{align*}
    y_i
    &=1_T \alpha_i  + L_T R_0 \eta_{1g} + L_TR_0\xi_{sg} + \varepsilon_{i}  \\
    &=1_T \alpha_i  + L_T \beta_{1g} + L_T\delta_{sg} + \varepsilon_{i}    , \quad \varepsilon_{i}\sim N(0,\Sigma_s)    
\end{align*}
where $\beta_{11,g}=0$ and $\delta_{s1,g}=0$. 
Let $\theta = (\eta_{1g},\xi_{sg},\Sigma_s,\phi_s,D_s)$ and $\alpha=\{\alpha_i\}.$
The likelihood conditional on the heterogeneous initial values $\alpha$ is
\begin{equation}
    f(y\vert\theta, \alpha) = 
    \prod_{g=1}^G 
    \prod_{i\in N_{1g} } 
    N_T\left( y_i \vert 1_T \alpha_i  + L_T R_0 \eta_{1g} , \Sigma_1\right)
    \cdot 
    \prod_{s\in\mathcal{S}}
    \prod_{g=1}^G
    \prod_{i\in N_{sg} }
    N_T\left( y_i \vert 1_T \alpha_i + L_T R_0 \eta_{1g} + L_TR_0 \xi_{sg}, \Sigma_s \right). \label{eq:like}    
\end{equation}
Integrating out $\alpha$,  we obtain 
\begin{equation}
    f(y\vert\theta) = 
    \prod_{g=1}^G 
    \prod_{i\in N_{1g} } 
    N_T\left( y_i \vert 1_T w_i'\phi_1  + L_T R_0 \eta_{1g} , \Lambda_1\right)
    \cdot 
    \prod_{s\in\mathcal{S}}
    \prod_{g=1}^G
    \prod_{i\in N_{sg} }
    N_T\left( y_i \vert 1_T w_i'\phi_s + L_T R_0 \eta_{1g} + L_TR_0 \xi_{sg}, \Lambda_s \right),\label{eq:like_int}    
\end{equation}
where for $s\in \{1\}\cup \mathcal{S}$, 
\[
    \Lambda_s = \Sigma_s  + D_s 1_T 1_T'.
\]

\subsection{Prior}
The priors are independently specified as follows. 
\begin{align*}
    \eta_{1g}&\overset{\text{ind}}{\sim} N_{T-1}(\mu_{\eta_{1g}},V_{\eta_{1g}}) , \hspace{2cm} g\in\mathcal{G} \\
    \xi_{sg}&\overset{\text{ind}}{\sim} N_{T-1}(\mu_{\xi_{sg}},V_{\xi_{sg}}), \hspace{2cm} g\in\mathcal{G}, s\in \mathcal{S} \\
    \sigma^2_{st}&\overset{\text{ind}}{\sim} \text{InvGam}(a_{st}/2, b_{st}/2) , \hspace{1.1cm} t=1\ldots,T, s\in \mathcal{S}\\
    \phi_s&\overset{\text{ind}}{\sim} N_K(\mu_{\phi_s},V_{\phi_s}) , \hspace{2.5cm} s\in \mathcal{S}\\
    D_s&\overset{\text{ind}}{\sim} \text{InvGam}(a_{D_s}/2,b_{D_s}/2) \hspace{1cm} s\in \mathcal{S}. 
\end{align*}
The prior on the parameters, especially the increment differences $\delta_{sg}$, needs to be carefully chosen, as the 
objects of interest,  the ATTs, are functions of these parameters. 
We recommend the following pre-training approach to set the priors. 
The idea is to use a random sample of the available data to train the prior hyperparameters, i.e., to 
estimate the hyperparameters, similar to an empirical Bayes approach. 
We randomly select $15\%$ of the units in each cohort-stratum cell for this purpose in the simulations and application below. 
Estimation and inference is then on the 
remaining data. 
\subsection{Posterior sampling}

A Gibbs sampler is used to efficiently sample the posterior distribution. 
Posterior inferences of objects of interest (ATTs) are based on the sample of draws produced by the algorithm. 
The posterior sample consists of 
\[
\theta^{(m)} = 
\left( 
\{ \eta_{1g}^{(m)},  g\in \mathcal{G}\}, 
\{ \xi_{sg}^{(m)},  s\in \mathcal{S},g\in \mathcal{G}\},
\{\Sigma_s^{(m)}, \phi_s^{(m)}, D_s^{(m)}: s\in \{1\}\cup \mathcal{S} \} \right), \ m=1,\ldots, M,
\]
where $M$ is the number of MCMC draws (beyond a suitable burn-in).
The algorithm sequentially samples from the conditional distributions  (see  Online Appendix  for details). 

\subsection{Frequentist validity of posterior inference}\label{sec:asymptotic_result}
We establish the frequentist validity of the posterior credible sets for the
ATT vector under a fixed-$T$, large-$n$ asymptotic framework. Let
$\theta=(\beta,\phi,\sigma^2,D)$ collect the model parameters and let
$\theta^\ast$ denote the true value. Define the ATT vector
\[
    \tau_{\mathrm{ATT}} = f(\theta)
    =
    \left\{ \tau_{\mathrm{ATT}}^{(s,g)}=L_s^{\text{post}}(\beta_{sg}-\beta_{1g})
    : s\in\mathcal S,\ g\in\mathcal G \right\}.
\]
Under the regularity conditions stated in Online Appendix,
the posterior distribution of $\theta$ satisfies a Bernstein--von Mises
theorem:
\[
    d_{\mathrm{TV}}\left[
    \Pi\left\{\sqrt n(\theta-\theta^\ast)\in\cdot\mid \mathrm{Data}_n\right\},
    N\left(\Delta_{n,\theta^\ast}, I_{\theta^\ast}^{-1}\right)
    \right]
    \overset{P_0}{\longrightarrow}0,
\]
where $I_{\theta^\ast}$ is the Fisher information matrix, $d_{\mathrm{TV}}$ is the total variation distance,  and
\[
    \Delta_{n,\theta^\ast}
    =
    \frac{1}{\sqrt n}
    \sum_{i=1}^n
    I_{\theta^\ast}^{-1}\ell_{\theta^\ast}(\mathrm{Data}_i),
\]
where $\ell_{\theta^\ast}(\cdot)$ is 
the individual score contribution.
Since $\tau_{\mathrm{ATT}}=f(\theta)$ is a smooth function of $\theta$,
the Bayesian delta method implies
\[
    \Pi\left\{
    \sqrt n(\tau_{\mathrm{ATT}}-\tau_{\mathrm{ATT}}^\ast)\in\cdot
    \mid \mathrm{Data}_n
    \right\}
    \rightsquigarrow
    N\left(
    \nabla f(\theta^\ast)' \Delta_{n,\theta^\ast},
    \nabla f(\theta^\ast)' I_{\theta^\ast}^{-1}\nabla f(\theta^\ast)
    \right),
\]
where $\tau_{\mathrm{ATT}}^\ast$ is the true ATT vector.  
Consequently, posterior credible sets for the ATT vector have asymptotically
correct frequentist coverage. See Online Appendix for proofs.

\section{Simulation and empirical application}\label{sec:simulation_application}
\subsection{Simulation}
We generate  data sets which, roughly speaking, mimic the real data that we use in the empirical application in Section \ref{sec:application}. The data set contains five periods (i.e. $T=5$), and units are treated in the second, fourth, and fifth periods (i.e. $\mathcal{S}=\{2,4,5\}$). 
Thus, there are 7 causal parameters of interest, for a given stratum. 
%We consider two sample sizes: a medium $n=250$ and moderately large $n=500$ values. 
The $n$ units are randomly allocated into the 4 treatment cohorts $\{1,2,4,5\}$ with probabilities $(0.4,0.2,0.2,0.2)$. 
We generated $w_{i}\sim \text{Unif}(L,U)$ to roughly match the moments of the log population in the real data in Section \ref{sec:application}.  
We set the data-generating parameter values as follows. First, we fit our model with no stratification (i.e. $G=1$) on the real data in Section \ref{sec:application}. When generating data under $G=1,2,$ and $3$, we set the true parameter values as 
$\delta_{s1}=\hat{\delta}_s$, 
$\delta_{s2}=c_2\hat{\delta}_s$, and 
$\delta_{s3}=c_3\hat{\delta}_s$, where 
$\hat{\delta}_s$ is the estimated value of $\delta_s$ above, and 
$c_2>1$ and $c_3<1$ are constants. 
The strata are defined based on  partitioning the interval $(L,U)$ into 2 and 3 sub-intervals of equal length under $G=2$ and $G=3$, respectively.

\subsubsection{Simulation 1: finite-sample advantages}
The goal of this simulation study is to demonstrate the finite-sample advantages of our proposed model when estimating strata-specific ATTs. We generate data based on three strata ($G=3$). Hence, there are 21 ATTs ($=7\times 3$). 

We compare three approaches: 
(1) our approach (BDID), 
(2) \cite{CallawaySantAnna2021did} applied separately on each stratum (CS21-Split), and 
(3) CS21 applied ignoring the strata (CS21-Pooled).  We use the R package, \texttt{did}, for CS21. 

For each approach, we report Bias, root-mean-squared-errors (RMSE), and empirical coverage of the 95\% confidence and credible intervals (Cov).
For CS-Pooled,  the same ATT estimates are used across all strata.

Table \ref{tab:sim_comparison_G3_prePT0_N100_R300} summarizes the results. Not surprisingly, CS21-Pooled performs substantially worse than the other two methods because it ignores treatment-effect heterogeneity across strata.

In contrast, the comparison between BDID and CS Split highlights the gains from the joint estimation across strata embedded in our joint estimation framework. Since CS21-Split estimates each stratum separately, it relies on  smaller effective sample sizes within each subgroup. BDID instead borrows strength across strata through the joint estimation of the full ATT array. As a result, BDID delivers substantially improved interval coverage and generally smaller bias in most ATT parameters, particularly for the larger treatment effects in Stratum 2.

Although CS Split occasionally attains slightly smaller RMSE values for some ATT parameters, this comes at the cost of severe undercoverage, with empirical coverage rates often far below the nominal 95\% level. In practice, this indicates that researchers can be overconfident about the strength of the effects even when the effective sample sizes for strata are too small to justify the underlying asymptotic approximation of the distribution of the estimator. Importantly, the finer the strata are, fewer observations there are to estimate strata-specific treatment effects. 

In contrast, BDID maintains coverage rates close to the nominal level across nearly all ATT parameters while keeping RMSE competitive. These results illustrate the finite-sample benefits of joint estimation and information sharing across strata in the proposed approach.

As expected, when $n$ becomes sufficiently large, asymptotic approximations become accurate within each stratum, and the performances of CS21-Split and BDID become increasingly similar. Additional simulation results for larger sample sizes are reported in the Online Appendix.

\FloatBarrier
\begin{table}[!htbp]
\centering
\caption{Simulation comparison: $G=3$, $n=100$}
\label{tab:sim_comparison_G3_prePT0_N100_R300}
\begin{threeparttable}
\resizebox{\textwidth}{!}{%
\begin{tabular}{lrrrrrrrrrr}
\hline\hline
ATT & True & BDID Bias & BDID RMSE & BDID Cov. & CS Split Bias & CS Split RMSE & CS Split Cov. & CS Pooled Bias & CS Pooled RMSE & CS Pooled Cov. \\
\hline
\multicolumn{11}{l}{\textit{Panel A: Individual ATT comparison}} \\
\hline
$\tau_{\mathrm{ATT}}(2,2; 1)$& -0.0209 & -0.0059 & 0.1200 & 0.960 & 0.0447 & 0.1075 & 0.849 & -0.0633 & 0.0888 & 0.770 \\
$\tau_{\mathrm{ATT}}(2,3; 1)$& -0.0810 & -0.0005 & 0.1152 & 0.963 & 0.0917 & 0.1277 & 0.736 & -0.2295 & 0.2538 & 0.390 \\
$\tau_{\mathrm{ATT}}(2,4; 1)$& -0.1434 & 0.0048 & 0.1231 & 0.963 & 0.0984 & 0.1457 & 0.756 & -0.4007 & 0.4351 & 0.263 \\
$\tau_{\mathrm{ATT}}(2,5; 1)$& -0.1074 & -0.0033 & 0.1264 & 0.957 & -0.0053 & 0.1069 & 0.920 & -0.3016 & 0.3279 & 0.317 \\
$\tau_{\mathrm{ATT}}(4,4; 1)$& 0.0035 & -0.0040 & 0.0884 & 0.997 & 0.0351 & 0.0807 & 0.860 & 0.0080 & 0.0450 & 0.947 \\
$\tau_{\mathrm{ATT}}(4,5; 1)$& -0.0350 & 0.0032 & 0.0958 & 0.990 & -0.0053 & 0.0800 & 0.920 & -0.0949 & 0.1126 & 0.593 \\
$\tau_{\mathrm{ATT}}(5,5; 1)$& -0.0271 & -0.0064 & 0.1448 & 0.967 & -0.0011 & 0.1253 & 0.877 & -0.0687 & 0.1008 & 0.813 \\
$\tau_{\mathrm{ATT}}(2,2; 2)$& -0.2091 & -0.0233 & 0.1363 & 0.940 & -0.0251 & 0.1053 & 0.870 & 0.1249 & 0.1395 & 0.453 \\
$\tau_{\mathrm{ATT}}(2,3; 2)$& -0.8104 & 0.0074 & 0.1447 & 0.970 & 0.0563 & 0.1171 & 0.840 & 0.4998 & 0.5115 & 0.003 \\
$\tau_{\mathrm{ATT}}(2,4; 2)$& -1.4337 & 0.0337 & 0.2090 & 0.953 & 0.0851 & 0.1337 & 0.783 & 0.8897 & 0.9057 & 0.000 \\
$\tau_{\mathrm{ATT}}(2,5; 2)$& -1.0738 & 0.0246 & 0.1817 & 0.943 & -0.0094 & 0.1188 & 0.880 & 0.6649 & 0.6772 & 0.000 \\
$\tau_{\mathrm{ATT}}(4,4; 2)$& 0.0349 & 0.0183 & 0.0986 & 0.970 & 0.1273 & 0.1482 & 0.607 & -0.0235 & 0.0502 & 0.900 \\
$\tau_{\mathrm{ATT}}(4,5; 2)$& -0.3503 & 0.0259 & 0.1116 & 0.960 & -0.0005 & 0.0857 & 0.907 & 0.2203 & 0.2285 & 0.040 \\
$\tau_{\mathrm{ATT}}(5,5; 2)$& -0.2713 & 0.0075 & 0.1543 & 0.963 & 0.0102 & 0.1266 & 0.893 & 0.1755 & 0.1903 & 0.297 \\
$\tau_{\mathrm{ATT}}(2,2; 3)$ & -0.0021 & -0.0189 & 0.1269 & 0.953 & 0.0421 & 0.1010 & 0.833 & -0.0821 & 0.1030 & 0.717 \\
$\tau_{\mathrm{ATT}}(2,3; 3)$ & -0.0081 & -0.0127 & 0.1267 & 0.953 & 0.0939 & 0.1306 & 0.689 & -0.3024 & 0.3213 & 0.133 \\
$\tau_{\mathrm{ATT}}(2,4; 3)$ & -0.0143 & -0.0130 & 0.1272 & 0.967 & 0.0931 & 0.1391 & 0.773 & -0.5297 & 0.5562 & 0.077 \\
$\tau_{\mathrm{ATT}}(2,5; 3)$ & -0.0107 & -0.0074 & 0.1354 & 0.937 & -0.0014 & 0.1090 & 0.900 & -0.3982 & 0.4185 & 0.087 \\
$\tau_{\mathrm{ATT}}(4,4; 3)$ & 0.0003 & -0.0010 & 0.0935 & 0.993 & 0.0349 & 0.0812 & 0.893 & 0.0111 & 0.0457 & 0.940 \\
$\tau_{\mathrm{ATT}}(4,5; 3)$ & -0.0035 & 0.0017 & 0.1099 & 0.993 & 0.0009 & 0.0967 & 0.867 & -0.1265 & 0.1402 & 0.407 \\
$\tau_{\mathrm{ATT}}(5,5; 3)$ & -0.0027 & -0.0004 & 0.1519 & 0.953 & -0.0041 & 0.1301 & 0.887 & -0.0932 & 0.1188 & 0.730 \\
\hline
\multicolumn{11}{l}{\textit{Panel B: Overall comparison}} \\
\hline
Method & -- & Mean Bias & Max $|$Bias$|$ & RMSE & SD & Coverage & -- & -- & -- & -- \\
BDID & -- & 0.0014 & 0.0337 & 0.1296 & 0.1467 & 0.964 & -- & -- & -- & -- \\
CS split & -- & 0.0363 & 0.1273 & 0.1141 & 0.0996 & 0.835 & -- & -- & -- & -- \\
CS pooled & -- & -0.0057 & 0.8897 & 0.2748 & 0.0927 & 0.423 & -- & -- & -- & -- \\
\hline\hline
\end{tabular}}
\begin{tablenotes}[flushleft]
\footnotesize
\item \textit{Notes:} BDID denotes the proposed estimator. CS Split and CS Pooled denote comparison estimators. 
\item Bias, RMSE, and coverage are computed over 300 simulation replications. 
\item Coverage refers to nominal 95\% credible/confidence interval coverage.
\item We generate data (and estimate the model for BDID) without  parallel trends condition in pre-treatment periods.
\end{tablenotes}
\end{threeparttable}
\end{table}

\FloatBarrier

\subsubsection{Simulation 2: asymptotic behavior}
The goal of this simulation is to confirm the large sample properties of posterior inference of our proposed approach given in  Section \ref{sec:asymptotic_result}. We generate data with two strata and $n\in\{250, 500, 1000\}$. 

Table \ref{tab:ATT_combined_G2_prePT0} summarizes the finite-sample performance of posterior inference under the proposed BDID framework. The results strongly support the asymptotic theory developed in Section \ref{sec:asymptotic_result}.

First, the posterior means exhibit negligible bias across all ATT parameters and sample sizes. The magnitude of the biases decreases further as $n$ increases, which is consistent with posterior consistency implied by the Bernstein--von Mises theorem.

Second, both RMSE and posterior standard deviations decrease steadily as the sample size increases. Importantly, the posterior standard deviations closely track the empirical RMSE values across all ATT parameters and sample sizes, indicating that the posterior distribution provides an accurate approximation to the true finite-sample sampling uncertainty.

Most importantly, the empirical coverage probabilities of the nominal 95\% posterior credible intervals remain close to the nominal level across all sample sizes. The overall coverage rates are 0.968, 0.958, and 0.955 for $n=250$, $500$, and $1000$, respectively. These findings provide strong empirical support for the asymptotic validity of posterior credible sets established in Section \ref{sec:asymptotic_result}.

Overall, the simulation results indicate that posterior inference under the proposed BDID framework performs reliably in finite samples 
and provides uncertainty quantification for ATTs with a valid frequentist interpretation.

\FloatBarrier
\begin{table}[htbp]
\centering
\caption{ATT estimation results across sample sizes: $G=2$
%prePT$_{DGP}=0$, prePT$_{EST}=0$
}
\label{tab:ATT_combined_G2_prePT0}
\normalsize
\begin{threeparttable}
\resizebox{\textwidth}{!}{%
\begin{tabular}{lrrrrrrrrrrrr}
\toprule
 & \multicolumn{3}{c}{Bias} & \multicolumn{3}{c}{RMSE} & \multicolumn{3}{c}{Post. SD} & \multicolumn{3}{c}{Coverage} \\
\cmidrule(lr){2-4} \cmidrule(lr){5-7} \cmidrule(lr){8-10} \cmidrule(lr){11-13}
ATT & $n=250$ & $n=500$ & $n=1000$ & $n=250$ & $n=500$ & $n=1000$ & $n=250$ & $n=500$ & $n=1000$ & $n=250$ & $n=500$ & $n=1000$ \\
\midrule
$\tau_{\mathrm{ATT}}(2,2; 1)$& -0.002 & 0.001 & -0.001 & 0.052 & 0.038 & 0.026 & 0.057 & 0.038 & 0.027 & 0.959 & 0.954 & 0.957 \\
$\tau_{\mathrm{ATT}}(2,3; 1)$& -0.001 & 0.002 & -0.001 & 0.047 & 0.034 & 0.024 & 0.053 & 0.036 & 0.025 & 0.974 & 0.954 & 0.956 \\
$\tau_{\mathrm{ATT}}(2,4; 1)$& -0.002 & 0.002 & -0.000 & 0.055 & 0.039 & 0.028 & 0.059 & 0.040 & 0.028 & 0.965 & 0.951 & 0.942 \\
$\tau_{\mathrm{ATT}}(2,5; 1)$& -0.002 & 0.001 & -0.001 & 0.059 & 0.040 & 0.028 & 0.060 & 0.041 & 0.028 & 0.957 & 0.953 & 0.949 \\
$\tau_{\mathrm{ATT}}(4,4; 1)$& 0.003 & 0.001 & 0.001 & 0.043 & 0.030 & 0.022 & 0.051 & 0.033 & 0.022 & 0.979 & 0.968 & 0.964 \\
$\tau_{\mathrm{ATT}}(4,5; 1)$& 0.002 & -0.002 & 0.000 & 0.044 & 0.033 & 0.023 & 0.053 & 0.035 & 0.024 & 0.986 & 0.964 & 0.950 \\
$\tau_{\mathrm{ATT}}(5,5; 1)$& 0.001 & 0.000 & 0.001 & 0.068 & 0.048 & 0.033 & 0.069 & 0.047 & 0.033 & 0.960 & 0.945 & 0.955 \\
$\tau_{\mathrm{ATT}}(2,2; 2)$& -0.002 & 0.000 & 0.000 & 0.053 & 0.036 & 0.027 & 0.057 & 0.039 & 0.027 & 0.963 & 0.963 & 0.951 \\
$\tau_{\mathrm{ATT}}(2,3; 2)$& 0.001 & 0.000 & -0.000 & 0.048 & 0.033 & 0.024 & 0.053 & 0.036 & 0.025 & 0.966 & 0.962 & 0.963 \\
$\tau_{\mathrm{ATT}}(2,4; 2)$& 0.001 & 0.001 & 0.000 & 0.055 & 0.037 & 0.027 & 0.059 & 0.040 & 0.028 & 0.965 & 0.965 & 0.946 \\
$\tau_{\mathrm{ATT}}(2,5; 2)$& 0.003 & 0.002 & 0.000 & 0.057 & 0.040 & 0.028 & 0.060 & 0.041 & 0.029 & 0.961 & 0.952 & 0.952 \\
$\tau_{\mathrm{ATT}}(4,4; 2)$& 0.001 & 0.002 & 0.001 & 0.043 & 0.030 & 0.021 & 0.051 & 0.033 & 0.023 & 0.981 & 0.968 & 0.965 \\
$\tau_{\mathrm{ATT}}(4,5; 2)$& 0.004 & 0.002 & 0.001 & 0.046 & 0.033 & 0.022 & 0.053 & 0.035 & 0.023 & 0.974 & 0.954 & 0.964 \\
$\tau_{\mathrm{ATT}}(5,5; 2)$& 0.004 & 0.001 & 0.001 & 0.065 & 0.046 & 0.032 & 0.069 & 0.047 & 0.033 & 0.964 & 0.958 & 0.952 \\
\midrule
OVERALL & 0.001 & 0.001 & 0.000 & 0.052 & 0.037 & 0.026 & 0.057 & 0.039 & 0.027 & 0.968 & 0.958 & 0.955 \\
\bottomrule
\end{tabular}
}
\begin{tablenotes}[flushleft]
\footnotesize
\item Notes: The table reports bias, RMSE, average posterior standard deviation, 
\item and empirical coverage of 95\% credible intervals across 1000 Monte Carlo repetitions.
\item We generate data and estimate the model without  parallel trends condition in pre-treatment periods.
\end{tablenotes}
\end{threeparttable}
\end{table}
\FloatBarrier

\subsubsection{Simulation 3: marginal likelihood and specification selection}
In Online Appendix, we report additional simulation evidence on the use of marginal likelihood comparison for specification assessment. The simulations examine both the ability of the marginal likelihood to distinguish models with and without pre-treatment parallel-trends restrictions conditional on a known structure of stratification, and its performance in comparing both the pre-trend restrictions as well as alternative stratification structures indexed by $G$. The results show that the marginal likelihood appropriately penalizes unnecessarily rich heterogeneity specifications while favoring more flexible models when genuine subgroup heterogeneity is present in the data.
%Section \ref{sec:prePT_ML} proposed marginal likelihood comparison as a probabilistic framework for evaluating alternative identifying assumptions and heterogeneity specifications. In this simulation study, we examine the finite-sample performance of this approach in two settings. Part A studies whether the marginal likelihood can correctly distinguish between models with and without pre-treatment parallel trends (\texttt{prePT}) restrictions when the degree of heterogeneity is fixed. This corresponds to settings in which the researcher already has a substantive stratification in mind. Part B then studies the more general problem of jointly selecting the number of strata $G$ and the pre-treatment specification. In particular, we examine whether the marginal likelihood appropriately penalizes unnecessarily rich heterogeneity structures while still detecting meaningful subgroup heterogeneity when supported by the data.

\subsection{Real data application: strata-specific ATTs of minimum wage policy on employment}\label{sec:application}
We illustrate the proposed approach on real data on  minimum wage policy and teen employment in the U.S. that we obtained from \cite{CallawaySantAnna2021did}\footnote{\url{https://bcallaway11.github.io/did/}}. This is a subset of the dataset analyzed in their paper. During the period between 2001 and 2007, the federal minimum wage was flat at \$5.15 per hour.  They focus on county-level teen employment in states where the minimum wage was equal to the federal minimum wage at the beginning of the period. Some of these states increased their minimum wage over this period. The counties in these states are the treated units. In particular, the treatment cohorts are defined by the time periods when the states first increased their minimum wage. Other states did not increase their minimum wage and the counties in these states are  never-treated. The data includes 500 counties (i.e.\ $n=500$): 309 counties were never-treated, 20 counties were treated in 2004, 40 were treated in 2006, and 131 were treated in 2007. The outcome variable is the log of county-level teen employment.  We use log population as a pre-treatment baseline covariate (i.e.\ $w_i$). The data spans from 2003 to 2007 (i.e.\ $T=5$). 
%A plot of the data on the outcome variable is given in Figure \ref{fig:datafig}. 
We construct the prior from 15\% of the data. This data is randomly selected from each cohort-stratum cell. The remaining 85\% of the data are used for estimation.

\begin{table}[htbp]
\centering
\caption{Distribution across cohorts and strata}
\label{tab:application_cohort_strata_dist}

\begin{threeparttable}
\begin{tabular}{c cc cc ccc}
\hline\hline
& \multicolumn{2}{c}{$g=1$ (small counties)} & \multicolumn{2}{c}{$g=2$ (large counties)} 
& \multicolumn{3}{c}{Log population} \\
\cline{2-3} \cline{4-5} \cline{6-8}
Cohort ($s$) & Counts & Percent & Counts & Percent 
& Min & Mean & Max \\
\hline
1 & 170 & 34.0 & 139 & 27.8 & 0.1765 & 3.1834 & 7.7048 \\
2 & 10  & 2.0  & 10  & 2.0  & 1.8060 & 3.4764 & 6.4683 \\
4 & 14  & 2.8  & 26  & 5.2  & 1.8388 & 3.7541 & 7.0310 \\
5 & 56  & 11.2 & 75  & 15.0 & 0.0658 & 3.4586 & 7.2399 \\
\hline
Total & 250 & 50.0 & 250 & 50.0 & -- & -- & -- \\
\hline\hline
\end{tabular}

\begin{tablenotes}
\footnotesize
\item Notes: Percent indicates the share of the full sample in each cohort--group cell. 
The threshold on log population is 3.2578.
\end{tablenotes}

\end{threeparttable}
\end{table}

When making policy recommendations, having information on heterogeneous treatment effects is important. In the context of the current application, it would be useful to learn how the effects of the minimum wage policy  differ between small and large counties. We hence consider partitioning the space of log population into two (i.e. $G=2$) so that  the numbers of counties in each cohort are roughly comparable  between strata. The first stratum consists of small counties and the second consists of large counties. 
See Table \ref{tab:application_cohort_strata_dist} for the structure of strata. 
Importantly, some of the cohort-stratum cells are sparse, casting doubt on the validity of the conventional asymptotic approximation of the distribution of the ATT estimator in the relevant cells (see Simulation 1). For example, the strata in Cohort 2 both have only 10 observations.

\FloatBarrier
\begin{figure}[htbp]
\centering

\begin{subfigure}[t]{0.32\textwidth}
    \centering
    \includegraphics[width=\textwidth]{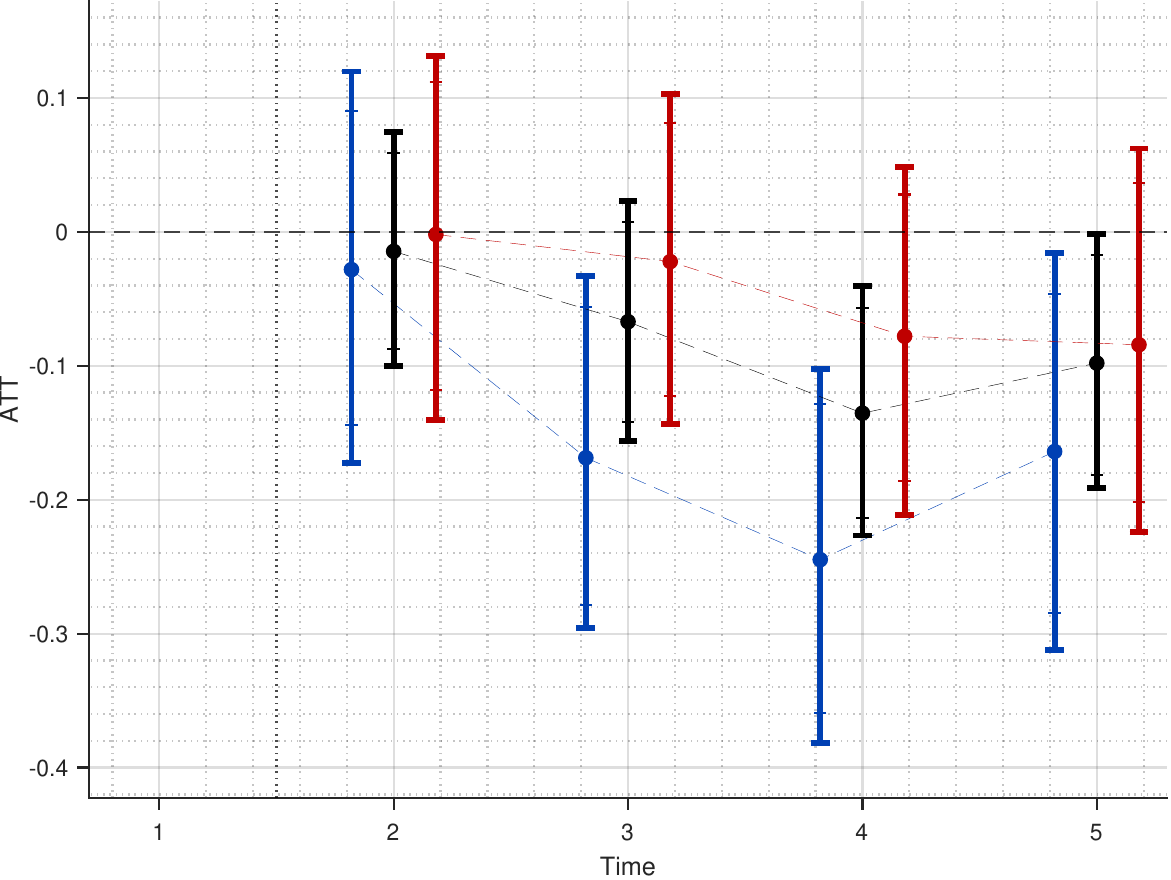}
    \caption{2004 Cohort ($s=2$)}
\end{subfigure}
\hfill
\begin{subfigure}[t]{0.32\textwidth}
    \centering
    \includegraphics[width=\textwidth]{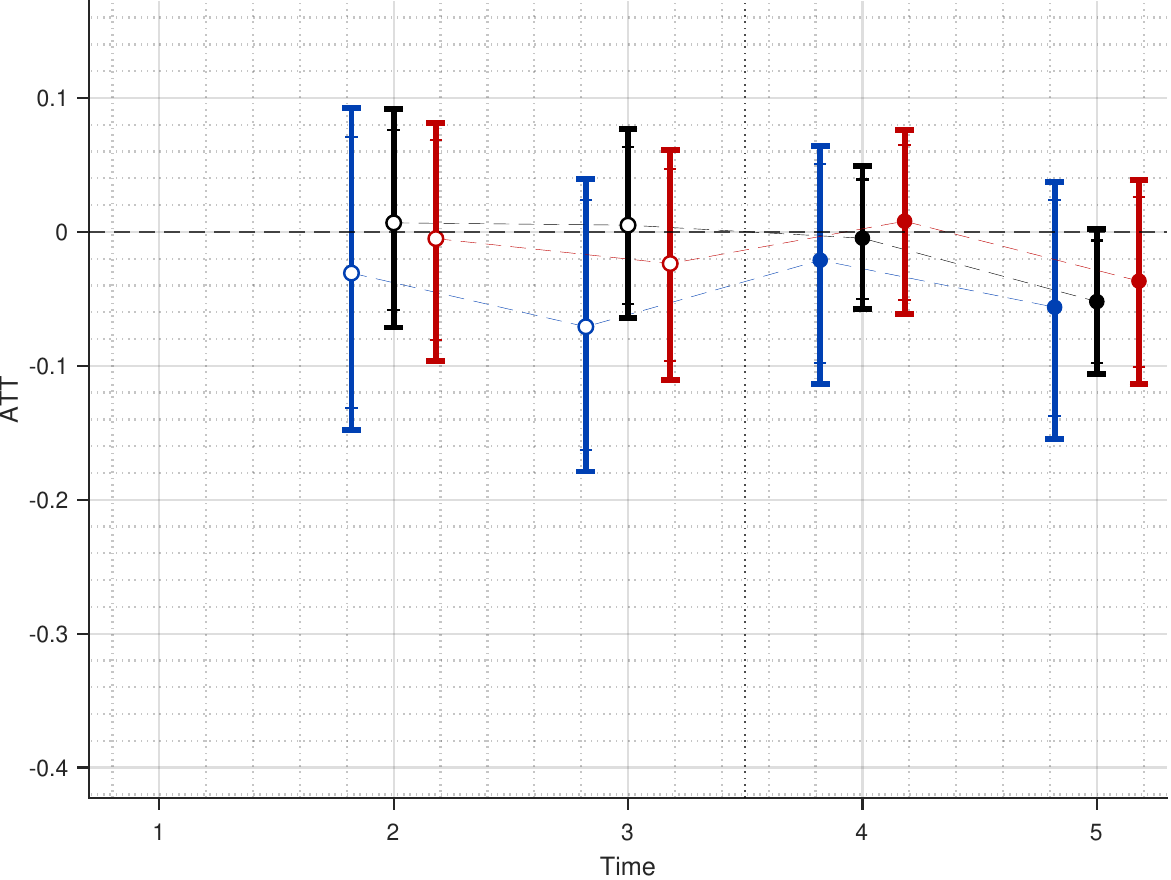}
    \caption{2006 Cohort ($s=4$)}
\end{subfigure}
\hfill
\begin{subfigure}[t]{0.32\textwidth}
    \centering
    \includegraphics[width=\textwidth]{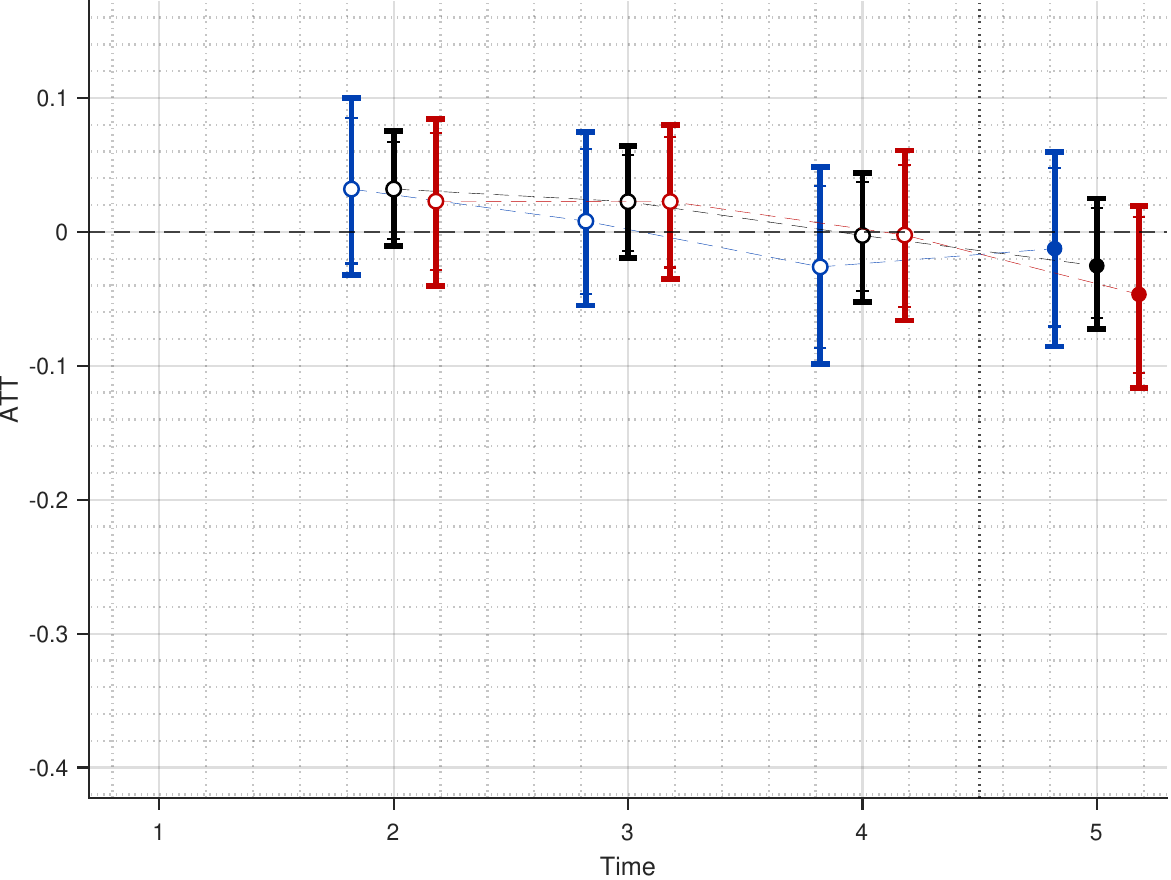}
    \caption{2007 Cohort ($s=5$)}
\end{subfigure}

\caption{
Posterior ATT estimates under $G=2$ and prePT$=0$.
Blue = stratum 1, red = stratum 2, black = $G=1$ benchmark.
Thick intervals are 95\% credible intervals and thin intervals are 90\% credible intervals.
Filled markers denote post-treatment effects and open markers denote pseudo-ATTs.
}
\label{fig:application_att_combined}

\end{figure}
\FloatBarrier

Figure \ref{fig:application_att_combined} visualizes  the estimated ATTs. 
With no stratification (black), there is a lagged effect in the 2004 cohort: no significant effect for the first two post periods, but 13.5\% and 9.8\% reduction in teen employment in the last two periods. These observations are  consistent with what one finds using CS21 in a usual fashion (i.e. without splitting  the data). 

Now, under the stratification with two clusters  ($G=2$), we see a  different story. For small counties ($g=1$), we found substantially larger effects for the 2004 cohort in the last two periods; 24.5\% and 16.4\%, respectively. On the other hand, we did not find significant effects among larger counties ($g=2$).  Hence, the strata-specific effects give an important policy implication that could have been ignored when focusing on the aggregate effects: the minimum wage increase can impact smaller counties more than larger counties. 

Table \ref{tab:application_all4models}
gives details of the estimation results. 
The log marginal likelihood under two strata was -65.6, and it was -48.7  when we impose the parallel trends in pre-periods, which supports our confidence in the parallel trend assumption. The posterior SD tends to be smaller under this restriction, reflecting the smaller effective number of parameters.

%The marginal likelihood is indeed the largest for the model without stratification and with the parallel trend restriction on the pre-periods with the two strata model under the pre-PT restriction being the second. Purely from the standpoint of marginal likelihood comparison, one may prefer the model with $G=1$ and \texttt{prePT}=1; however, practically speaking, the policy implication obtained by learning the strata-specific heterogeneous effects could justify the choice of the stratification into two clusters. 

From the perspective of marginal likelihood comparison, the specification with $G=1$ and \texttt{prePT}=1 receives the strongest support from the data. Nevertheless, the $G=2$ specification uncovers substantively important heterogeneity between small and large counties that may be relevant for policy analysis. This illustrates that model comparison can inform, but need not fully determine, the degree of heterogeneity retained in empirical applications.

\begin{table}[htbp]
\centering
\caption{BDID estimates across strata and pre-trend specifications}
\label{tab:application_all4models}
\scriptsize
\begin{threeparttable}
\resizebox{\textwidth}{!}{%
\begin{tabular}{lrrrrrrrrrrrrrrrr}
\hline\hline
ATT & \multicolumn{4}{c}{$G=2$, prePT=0} & \multicolumn{4}{c}{$G=2$, prePT=1} & \multicolumn{4}{c}{$G=1$, prePT=0} & \multicolumn{4}{c}{$G=1$, prePT=1} \\
 & Mean & SD & LB & UB & Mean & SD & LB & UB & Mean & SD & LB & UB & Mean & SD & LB & UB \\
\hline
$\tau_{\mathrm{ATT}}(2,2; 1)$& -0.028 & 0.072 & -0.173 & 0.120 & -0.035 & 0.068 & -0.163 & 0.106 & -0.015 & 0.045 & -0.100 & 0.074 & -0.024 & 0.050 & -0.123 & 0.077 \\
$\tau_{\mathrm{ATT}}(2,3; 1)$& -0.169 & 0.068 & -0.296 & -0.033 & -0.167 & 0.061 & -0.281 & -0.045 & -0.067 & 0.046 & -0.156 & 0.023 & -0.073 & 0.045 & -0.161 & 0.015 \\
$\tau_{\mathrm{ATT}}(2,4; 1)$& -0.245 & 0.069 & -0.382 & -0.102 & -0.230 & 0.063 & -0.351 & -0.104 & -0.135 & 0.048 & -0.227 & -0.040 & -0.129 & 0.047 & -0.217 & -0.034 \\
$\tau_{\mathrm{ATT}}(2,5; 1)$& -0.164 & 0.073 & -0.312 & -0.016 & -0.160 & 0.068 & -0.293 & -0.029 & -0.098 & 0.049 & -0.191 & -0.001 & -0.102 & 0.050 & -0.197 & -0.003 \\
$\tau_{\mathrm{ATT}}(4,4; 1)$& -0.021 & 0.045 & -0.113 & 0.064 & -0.029 & 0.033 & -0.095 & 0.034 & -0.005 & 0.027 & -0.057 & 0.049 & 0.006 & 0.020 & -0.032 & 0.045 \\
$\tau_{\mathrm{ATT}}(4,5; 1)$& -0.056 & 0.050 & -0.155 & 0.037 & -0.074 & 0.039 & -0.148 & 0.003 & -0.052 & 0.028 & -0.106 & 0.002 & -0.050 & 0.023 & -0.094 & -0.003 \\
$\tau_{\mathrm{ATT}}(5,5; 1)$& -0.012 & 0.037 & -0.086 & 0.060 & -0.043 & 0.029 & -0.101 & 0.015 & -0.025 & 0.025 & -0.072 & 0.025 & -0.044 & 0.020 & -0.084 & -0.003 \\
$\tau_{\mathrm{ATT}}(2,2; 2)$& -0.002 & 0.071 & -0.141 & 0.131 & -0.007 & 0.065 & -0.131 & 0.122 & -- & -- & -- & -- & -- & -- & -- & -- \\
$\tau_{\mathrm{ATT}}(2,3; 2)$& -0.022 & 0.063 & -0.144 & 0.103 & -0.030 & 0.059 & -0.147 & 0.089 & -- & -- & -- & -- & -- & -- & -- & -- \\
$\tau_{\mathrm{ATT}}(2,4; 2)$& -0.078 & 0.066 & -0.211 & 0.049 & -0.067 & 0.061 & -0.186 & 0.057 & -- & -- & -- & -- & -- & -- & -- & -- \\
$\tau_{\mathrm{ATT}}(2,5; 2)$& -0.084 & 0.074 & -0.224 & 0.062 & -0.088 & 0.069 & -0.229 & 0.045 & -- & -- & -- & -- & -- & -- & -- & -- \\
$\tau_{\mathrm{ATT}}(4,4; 2)$& 0.008 & 0.036 & -0.061 & 0.076 & 0.009 & 0.026 & -0.040 & 0.061 & -- & -- & -- & -- & -- & -- & -- & -- \\
$\tau_{\mathrm{ATT}}(4,5; 2)$& -0.037 & 0.039 & -0.114 & 0.039 & -0.047 & 0.032 & -0.110 & 0.014 & -- & -- & -- & -- & -- & -- & -- & -- \\
$\tau_{\mathrm{ATT}}(5,5; 2)$& -0.047 & 0.034 & -0.116 & 0.019 & -0.064 & 0.027 & -0.116 & -0.009 & -- & -- & -- & -- & -- & -- & -- & -- \\
\hline
Log ML & \multicolumn{4}{c}{-65.651} & \multicolumn{4}{c}{-48.716} & \multicolumn{4}{c}{-54.171} & \multicolumn{4}{c}{-41.852} \\
\hline\hline
\end{tabular}}
\begin{tablenotes}[flushleft]
\footnotesize
\item Notes: Posterior means, posterior standard deviations, and 95\% credible intervals from BDID estimation.
\item Columns compare four specifications: two strata choices, $G=1,2$, and two pre-trend restrictions, prePT$=0,1$.
\item Log ML denotes the estimated log marginal likelihood.
\end{tablenotes}
\end{threeparttable}
\end{table}
\FloatBarrier

\section{Conclusion}\label{sec:conclusion}
This paper develops a Bayesian framework for estimating heterogeneous treatment effects in staggered Difference-in-Differences designs with covariate-defined strata. By jointly modeling cohort-, time-, and stratum-specific ATT parameters within a unified likelihood-based framework, the approach stabilizes inference in settings with sparse cohort-by-time-by-stratum cells while preserving the estimand structure of modern staggered DiD methods. 

The framework also provides a probabilistic approach to assessing the validity of the parallel trends condition in pre-treatment periods through marginal likelihood comparison. Rather than conditioning inference on conventional pre-trend tests, the proposed approach incorporates uncertainty regarding pre-treatment restrictions directly into posterior inference. The same framework can be used to assess how much subgroup heterogeneity is supported by the data. 

We establish a Bernstein--von Mises theorem for the ATT array, implying that posterior credible sets possess asymptotically valid frequentist coverage. Simulation studies demonstrate substantial finite-sample improvements relative to existing approaches applied separately within strata, particularly in settings with limited subgroup sample sizes. In the empirical application, the framework uncovers important heterogeneity in the employment effects of minimum wage increases across small and large counties that is obscured in aggregate analyses. 

Several extensions are possible. Future work could study fully data-driven stratification rules, more flexible outcome dynamics, and semiparametric or nonparametric specifications for the outcome distribution. More broadly, the framework illustrates how Bayesian joint modeling can provide a useful complement to modern Difference-in-Differences methods in empirically relevant high-dimensional heterogeneous treatment effect settings.

%Our approach is implemented in a user-friendly software package, \texttt{bdid}, to support its use by practitioners. It is available for MATLAB, R, and Stata.

\vspace{1cm}
\begin{spacing}{1.05}
\bibliographystyle{apacite}

\bibliography{references_did}
\end{spacing}

%\iffalse 
\begin{center}
{\large\bf Acknowledgments}
\end{center}
We benefited from helpful comments from  
Matt Webb, 
Jun Zhao, 
Irene Botosaru, 
and  
Gonzalo Vazquez-Bare.
%\fi 

\clearpage

% --- Online Appendix ---
%\iffalse 
\renewcommand{\thesection}{A\arabic{section}}
\renewcommand{\theequation}{A\arabic{equation}}
\renewcommand{\thetable}{A\arabic{table}}
\renewcommand{\thefigure}{A\arabic{figure}}
\setcounter{section}{0}
\setcounter{equation}{0}
\setcounter{table}{0}
\setcounter{figure}{0}

%\pagenumbering{arabic} % or roman if preferred
%\setcounter{page}{1}
\setcounter{section}{0}
\setcounter{equation}{0}
\setcounter{table}{0}
\setcounter{figure}{0}
\renewcommand{\thesection}{S\arabic{section}}
\renewcommand{\theequation}{S\arabic{equation}}
\renewcommand{\thetable}{S\arabic{table}}
\renewcommand{\thefigure}{S\arabic{figure}}
%\fi

\appendix 
\begin{center}
{\bf\Large Online Appendix}\\[1ex]
\end{center}
\onehalfspacing
Section \ref{sec:mcmcsteps} describes the MCMC steps and 
Section \ref{sec:prePT_appdix} illustrates the model under the restriction that the parallel trends hold in pre-treatment periods. 
Section \ref{sec:theory} presents the proofs for the asymptotic justification of posterior inference. 
Section \ref{sec:additional_simulation} provides additional simulation results. 
\section{MCMC steps}\label{sec:mcmcsteps}
\vspace*{0.75\baselineskip}
\subsection{Gibbs sampling steps}
\subsubsection{Sample $\eta_1$ and $\xi$ from the marginalized model}
For $i\in N_{sg}$, define the residuals from the marginalized model:   
\[
\tilde{e}_i(\beta_1,\delta,\phi) = y_i  -1_T w_i'\phi_s - L_T  \beta_{1g} -1(s\ne 1) L_T  \delta_{sg}
\]
and let $\tilde{e}_{it}(\beta_1,\delta,\phi) $ be its $t$th element. 
\paragraph{Update $\eta_{1g}$}
\[
\eta_{1g} \vert \bullet \sim N_{T-1}(m_{1g},B_{1g}^{-1}),
\]
where 
\begin{align*}
    B_{1g}&=V_{\eta_{1g}}^{-1} + \sum_{s\in \{1\}\cup\mathcal{S}} 
    n_{sg} R_0' L_T' \Lambda_s^{-1} L_T R_0\\
    m_{1g}&=B_{1g}^{-1}\left(V_{\eta_{1g}}^{-1} \mu_{\eta_{1g}} +\sum_{s\in \{1\}\cup\mathcal{S}}  \sum_{i\in N_{sg}} R_0' L_T'\Lambda_s^{-1} \tilde{e}_i(\beta_1=0,\delta,\phi)\right)
\end{align*}
\paragraph{Update $\xi_{sg}$}
\[
\xi_{sg} \vert \bullet \sim N_{T-1}(m_{sg},B_{sg}^{-1}),
\]
where 
\begin{align*}
    B_{sg}&=V_{\xi_{sg}}^{-1} +  n_{sg} R_0'L_T' \Lambda_s^{-1} L_TR_0\\
    m_{sg}&=B_{sg}^{-1}\left(V_{\xi_{sg}}^{-1} \mu_{\xi_{sg}} +  \sum_{i\in N_{sg}} R_0'L_T'\Lambda_s^{-1} \tilde{e}_i(\beta_1,\delta=0,\phi)\right)
\end{align*}
\subsubsection{Sample $\alpha$ and $\sigma^2$ from the conditional model}
For $i\in N_{sg}$, define the residuals from the conditional model:  
\[
e_i(\beta_1,\delta,\alpha) = y_i  -1_T \alpha_i - L_T  \beta_{1g} -1(s\ne 1) L_T  \delta_{sg}
\]
and let $e_{it}(\beta_1,\delta,\alpha) $ be its $t$th element. 
\paragraph{Update $\alpha_i$}
\[
\alpha_i \vert \bullet \sim N(m_i,B_i^{-1}),
\]
where, for $s=S_i$,
\begin{align*}
    B_{i}&=D_s^{-1} +  1_T' \Sigma_{s}^{-1} 1_T\\
    m_{i}&=B_{i}^{-1}\bigg(D_s^{-1} (w_i'\phi_s) +  1_T'\Sigma_{s}^{-1} e_i(\beta_1,\delta,\alpha=0)\bigg)
\end{align*}
\paragraph{Update $\sigma^2_{st}$}
\[
\sigma^2_{st} \vert \bullet \sim \text{InvGam}(\bar{a}_{st}/2, \bar{b}_{st}/2),
\]
where 
\begin{align*}
    \bar{a}_{st}&=a_{st} + \sum_g n_{sg}=a_{st} +  n_{s}\\
    \bar{b}_{st}&=b_{st} + \sum_g \sum_{i\in N_{sg}}  e^2_{it}(\eta_1,\xi,\alpha)
\end{align*}
\subsubsection{Sample RE parameters}

\paragraph{Update $\phi_s$}
\[
\phi_s \vert \bullet \sim N_K(m_{\phi_s},B_{\phi_s}^{-1}),
\]
where 
\begin{align*}
     B_{\phi_s}&=V_{\phi_s}^{-1} + D_s^{-1}\sum_{i:S_i=s} w_i w_i',\\
     m_{\phi_s}&=B_{\phi_s}^{-1}
     \left(
     V_{\phi_s}^{-1} \mu_{\phi_s}      
     + 
     D_s^{-1}
     \sum_{i:S_i=s}
     w_i  \alpha_i 
     \right)
\end{align*}
\paragraph{Update $D_s$}
\[
D_s \vert \bullet \sim \text{InvGam}(\bar{a}_{D_s}/2, \bar{b}_{D_s}/2),
\]
where 
\begin{align*}
    \bar{a}_{D_s}&=a_{D_s} + n_s\\
    \bar{b}_{D_s}&=b_{D_s} +  \sum_{i:S_i=s}  (\alpha_i - w_i'\phi_s)^2
\end{align*}

%%%%%%%%%%%%%%%%%%%%%%%%%%%%%%%%%%%%%%%%%%%%%%%%%%%%%%%%%%%%%%%%%%%%%%%%%%%%%%%%%%%%%%%%%%%%%%%%%%%%%%%%%%%%
% prePT
%%%%%%%%%%%%%%%%%%%%%%%%%%%%%%%%%%%%%%%%%%%%%%%%%%%%%%%%%%%%%%%%%%%%%%%%%%%%%%%%%%%%%%%%%%%%%%%%%%%%%%%%%%%%
\section{Parallel trends imposed on the pre-periods}\label{sec:prePT_appdix}
We call this restriction `pre-PT.' In this case, we still have, for each stratum $g\in\mathcal{G}$,
for $i\in N_{1g}$, 
\begin{align*}
    y_i
    &=1_T \alpha_i + L_T R_0 \eta_{1g} + \varepsilon_{i}\\
    &=1_T \alpha_i + L_T \beta_{1g} + \varepsilon_{i}, \quad \varepsilon_{i}\sim N(0,\Sigma_1)  
\end{align*}
and for $i \in N_{sg}, s\in \mathcal{S}$,
\begin{align*}
    y_i
    &=1_T \alpha_i  + L_T R_0 \eta_{1g} + L_TR_0\xi_{sg} + \varepsilon_{i}  \\
    &=1_T \alpha_i  + L_T \beta_{1g} + L_T\delta_{sg} + \varepsilon_{i}    , \quad \varepsilon_{i}\sim N(0,\Sigma_s)    
\end{align*}
where $\beta_{11,g}=0$ and $\delta_{s1,g}=0$ as before, 
but now with the additional restriction that $\delta_{st,g}=0$ for the pre-periods $t<s$. 

Under this restriction, there are $T-s+1$ effective parameters in $\delta_{sg}$. 
Hence, define the $(T-s+1)$-dimensional vector $\gamma_{sg}=(\delta_{sg,s},\delta_{sg,s+1},\ldots,\delta_{sg,T})'$. 
It is related to the original vector via $\delta_{sg}=R_s\gamma_{sg}$, 
where  $R_s=[e_s,\ldots,e_T]$ is the $T\times (T-s+1)$ selection matrix. 
We still have $\beta_{1g}=R_0 \eta_{1g}$ as in the baseline case.

We can then write the restricted model for the treated cohorts as  
\[
    y_i=1_T \alpha_i  +L_T R_0\eta_{1g} + L_T R_s \gamma_{sg} + \varepsilon_{sg}    , \quad \varepsilon_{sg}\sim N(0,\Sigma_s), \quad s\in\mathcal{S},
\]
Hence the conditional likelihood is
\begin{equation}
    f(y\vert\theta,\alpha) = 
    \prod_{g=1}^G 
    \prod_{i\in N_{1g} } 
    N_T\left( y_i \vert 1_T \alpha_i  + L_T R_0\eta_{1g} , \Sigma_1\right)
    \cdot 
    \prod_{s\in\mathcal{S}}
    \prod_{g=1}^G
    \prod_{i\in N_{sg} }
    N_T\left( y_i \vert 1_T \alpha_i  + L_T R_0\eta_{1g} + L_T R_s \gamma_{sg}, \Sigma_s \right). \label{eq:like}    
\end{equation}
Integrating out $\alpha$, 
\begin{equation}
    f(y\vert\theta) = 
    \prod_{g=1}^G 
    \prod_{i\in N_{1g} } 
    N_T\left( y_i \vert 1_T w_i'\phi_1  + L_T R_0\eta_{1g} , \Lambda_1\right)
    \cdot 
    \prod_{s\in\mathcal{S}}
    \prod_{g=1}^G
    \prod_{i\in N_{sg} }
    N_T\left( y_i \vert 1_T w_i'\phi_s + L_T R_0\eta_{1g} + L_T R_s \gamma_{sg}, \Lambda_s \right),\label{eq:like_int}    
\end{equation}
where for $s\in \{1\}\cup \mathcal{S}$, 
\[
    \Lambda_s = \Sigma_s  + D_s 1_T 1_T'.
\]
The  priors are given as
\begin{align*}
    \eta_{1g}&\overset{\text{ind}}{\sim} N_{T-1}(\mu_{\eta_{1g}},V_{\eta_{1g}}) , \hspace{2cm} g\in\mathcal{G} \\
    \gamma_{sg}&\overset{\text{ind}}{\sim} N_{T-s+1}(\mu_{\gamma_{sg}},V_{\gamma_{sg}}), \hspace{1.6cm} g\in\mathcal{G}, s\in \mathcal{S} \\
    \sigma^2_{st}&\overset{\text{ind}}{\sim} \text{InvGam}(a_{st}/2, b_{st}/2) , \hspace{1.1cm} t=1\ldots,T, s\in \mathcal{S}\\
    \phi_s&\overset{\text{ind}}{\sim} N_K(\mu_{\phi_s},V_{\phi_s}) , \hspace{2.5cm} s\in \mathcal{S}\\
    D_s&\overset{\text{ind}}{\sim} \text{InvGam}(a_{D_s}/2,b_{D_s}/2) \hspace{1cm} s\in \mathcal{S}. 
\end{align*}
\subsection{Sampling under pre-PT}
The sampling of $\eta_{1g}, \sigma^2_{st}, \alpha_i, \phi_s,$ and $D_s$ stay unchanged once sampling $\gamma_{sg}$  and replacing $\delta_{sg}$ by $R_s\gamma_{sg}$ in all expressions. 
For $i\in N_{sg}$, let  
\[
\tilde{e}_i(\beta_1,\delta(\gamma),\phi) = y_i  -1_T w_i'\phi_s - L_T \beta_{1g} -1(s\ne 1) L_T \delta_{sg}
\]
where $\delta_{sg} = R_s \gamma_{sg}$ 
and $\tilde{e}_{it}(\beta_1,\delta,\phi) $ be its $t$th element. 

\paragraph{Update $\gamma_{sg}$} 
\[
\gamma_{sg} \vert \bullet \sim N_{T-s+1}(m_{sg},B_{sg}^{-1}),
\]
where 
\begin{align*}
    B_{sg}&=V_{\gamma_{sg}}^{-1} +  n_{sg} R_s'L_T' \Lambda_s^{-1} L_T R_s\\
    m_{sg}&=B_{sg}^{-1}\left(V_{\gamma_{sg}}^{-1} \mu_{\gamma_{sg}} +  \sum_{i\in N_{sg}} R_s'L_T'\Lambda_s^{-1} \tilde{e}_i(\beta_1,\delta=0,\phi)\right)
\end{align*}

\section{Theory}\label{sec:theory}
In this section, we derive the frequentist properties of the proposed Bayesian approach. In particular, we establish a Bernstein--von Mises type theorem for the ATTs. This indicates that the Bayesian credible sets of the ATTs, which can be obtained easily from the draws of the simple MCMC algorithm, have frequentist interpretations. 

For  notational simplicity, we establish theoretical results under the original parametrization. All the results apply to the reparametrized model. 
Let $\beta$ and $\phi$ be the vectors with elements 
$\{\beta_{sg}\}_{s \in \{1\} \cup \mathcal{S}, g\in \mathcal{G}}$ 
and 
$\{\phi_s\}_{s \in \{1\} \cup \mathcal{S}}$, respectively. 
We impose the normalization that we discussed, namely $\beta_{1s,g}=0$.
Let      
$\sigma^2$ and $D$  
be the vectors with elements 
$\{ (\sigma^2_{s1},\ldots,\sigma^2_{sT}) \}_{s \in \{1\} \cup \mathcal{S}}$ 
and 
$\{ D_{s} \}_{s \in \{1\} \cup \mathcal{S}}$, respectively. 
Recall that $\Sigma_s = \text{diag}((\sigma^2_{s1},\ldots,\sigma^2_{sT}))$. 
For this section, we define the vector that collects all the parameters as   
$\theta=( \beta, \phi,  \sigma^2, D) \in \theta$. Let $\theta^*$ be the true value of $\theta$. 

We consider the asymptotic framework with a fixed $T$ and an increasing $n$. Conditional on covariates $w_i \in \mathbb{R}^{K}$ (and hence stratum assignment $G_i=g\in\mathcal{G}$) and treatment assignment $S_i=s\in \{1\} \cup \mathcal{S}$,  a sequence of outcomes $y_i=(Y_{i1},\ldots,Y_{iT})'\in \mathbb{R}^T$ is generated from the model $p_{\theta^*}$ defined as: 
\[
p_{\theta}(y_i \vert w_i, S_i=s, G_i=g)=N(y_i \vert  X_i  \psi_{sg}, \Lambda_s),
\]
where $  X_i=(1_T w_i', L_T)$, $\psi_{sg}=(\phi_s',\beta_{sg}')'$, and $\Lambda_s = \Sigma_s + D_s 1_T 1_T'$. The data contains outcomes $y_i$, covariates $w_i$, and treatment assignment $S_i$: $\text{Data}^n=\{\text{Data}_i=(y_i, w_i, S_i): i=1,\ldots,n\}$. The covariates $w_i$ are iid and generated  from $q^*$. The $G$ partitions of the compact covariate space $\mathcal{G}$ ensures $r^*_g=\Pr(G_i=g)>0$ for all $g\in \mathcal{G}$. 

The treatment assignments are iid and generated from $h^*=\{h^*_s: s\in \{1\} \cup \mathcal{S} \}$ with $h^*_s=\Pr(S_i=s)>0$ for  $\forall s$. The positivity of $h^*_s$ is ensured by the definition of $\mathcal{S}$ and the assumption that the never-treated units exist. We do not model $q^*$ and $h^*$.
The joint probability measure implied by $p_{\theta^*}$, $q^*$, and $h^*$ is denoted by $F_0$. We first establish  Bernstein--von Mises theorem for the model parameters $\theta$. 

\begin{theorem}
    Suppose 
    (1) $\theta = \Theta_\beta \times \Theta_\phi \times \Theta_\sigma \times \Theta_D$, where 
    $\Theta_\beta \subset \mathbb{R}^{TG(1+\vert \mathcal{S}\vert)}$ and 
    $\Theta_\phi \subset \mathbb{R}^{K(1+\vert \mathcal{S}\vert)}$ are compact,  
    $\Theta_\sigma = [\ell_\sigma,u_\sigma]^{T(1+\vert \mathcal{S}\vert)}$, and 
    $\Theta_D = [\ell_D,u_D]^{1+\vert \mathcal{S}\vert}$, 
    for some $0<\ell_\sigma,\ell_D$ and $u_\sigma, u_D <\infty$, 
    (2) $\mathcal{W}\subset \mathbb{R}^K$ is compact, 
    $\int \Vert X'X \Vert q^*(w) dw < \infty$ and $\mathbb{E}(X' \Lambda^{-1}_s X)$ is positive definite for all $\Sigma_s,D_s$ for all $s\in \{1\} \cup \mathcal{S}$, and 
    (3) $\theta^* \in \text{int}(\Theta)$. 
    Suppose that the prior is absolutely continuous in a neighborhood of $\theta^*$ with a continuous positive density at $\theta^*$. Let $\hat{\theta}_n$ be the maximum likelihood estimator of $\theta$ and $I_{\theta^*}$ be the Fisher information matrix at $\theta^*$. Then, as $n \to \infty$,
    \[
    d_{TV}\bigg( \pi\left( \sqrt{n}(\theta - \theta^*) \vert \text{Data}^n \right), 
    N\left(\Delta_{n,\theta^*}, I^{-1}_{\theta^*} \right)\bigg) \to 0
    \]
    in $F^\infty_0$-probability, where $\Delta_{n,\theta^*}=\frac{1}{\sqrt{n}}\sum_{i=1}^n  I^{-1}_{\theta^*} \ell_{\theta^*(\text{Data}_i)}$ with $\ell_{\theta}$ being the score function. 
\end{theorem}
\begin{proof}
    Lemmas 1–2 verify the identifiability and continuity assumptions required for the existence of uniformly consistent tests; see Theorem 6.6 of van der Vaart (2000). Since the model is finite dimensional and differentiable in quadratic mean with nonsingular Fisher information, Theorem 10.1 of van der Vaart (2000) implies the Bernstein–von Mises result (\citealp{van2000asymptotic}, ch.10), together with the prior positivity condition. 
\end{proof}
Now we study the asymptotic behavior of the posterior distribution of the ATTs. Let $P$ denote the prior for $\theta$ and $\theta_n^P$ denote the random variable with the law equal to the posterior distribution of $\theta$ given a sample $\text{Data}^n$ of size $n$. By Theorem 1, it converges in distribution to $N(\theta^* + \frac{1}{\sqrt{n}}\Delta_{n,\theta^*}, (n I_{\theta^*})^{-1} )$. Let $\tau_{ATT}$ be the vector with elements $\{ \tau_{ATT}^{(sg)} \}_{s\in \mathcal{S},\ g\in\mathcal{G}}$, where $ \tau_{ATT}^{(sg)} = L^{post} (\beta_{sg} - \beta_{1g})$. Denote the mapping $\theta \mapsto \tau_{ATT}$ by $f$; i.e. $f(\theta) = \tau_{ATT}$. Define the random variable $\tau_{ATT,n}^P=f(\theta_n^P)$ and, similarly, $\hat{\tau}_{ATT, n} = f(\hat{\theta}_n)$, where $\hat{\theta}_n$ is a consistent estimator for $\theta^*$. Note that the plug-in estimator is consistent, i.e, $\hat{\tau}_{ATT, n} \overset{p}{\to} \tau_{ATT}^*$, the true value of $\tau_{ATT}$.

By applying the Bayesian delta-method %(e.g.\ \citealp{BernardoSmith1994bayesian}, Section 5.3), 
we obtain the following as a consequence of Theorem 1. 
\begin{cor}
    Suppose the conditions of Theorem 1. Then, as $n \to \infty$,  
    $\tau_{ATT,n}^P$ converges in distribution to 
    \[
    N\left( f( \tilde{\theta}_n ), 
    \nabla f(\tilde{\theta}_n)' (n I_{\theta^*})^{-1} \nabla f(\tilde{\theta}_n) \right)    
    \]
    where $\tilde{\theta}_n = \theta^* + \frac{1}{\sqrt{n}}\Delta_{n,\theta^*} \overset{p}{\to} \theta^* $ and $\nabla$ denotes the gradient operator. 
\end{cor}

In other words, the posterior of ATTs is asymptotically normal. Therefore, Bayesian credible sets have asymptotically correct nominal coverage and are valid confidence sets. 
The priors introduced in Section %\ref{sec:prior}
satisfy the assumptions underlying the theoretical results.

\subsection{Proofs of the intermediate results} \label{sec:additional_theory}
We use these results in the proof of Theorem 1. 
\begin{lemma}[Identifiability]
    Suppose that $\mathbb{E}(X' \Lambda^{-1}_s X)$ is positive definite for all $\Sigma_s,D_s$ for all $s\in \{1\} \cup \mathcal{S}$. Then, $p_{\theta_1}(y \vert w) \ne p_{\theta_2}(y \vert w)$  with positive probability for all $\theta_1 \ne \theta_2$.
\end{lemma}
Let $\Vert \cdot \Vert_1$ and $\Vert \cdot \Vert_2$ be the $L_1$ and $L_2$ norms, respectively. 
Define the distance between parameters by 
    \[
    d(\theta_1, \theta_2)=
    \max 
    \left\{
    \Vert \beta_1 - \beta_2 \Vert_2^2 ,   
    \Vert \phi_1 - \phi_2 \Vert_2^2,    
    \Vert \sigma^2_1 - \sigma^2_2 \Vert_1,
    \Vert D_1 - D_2 \Vert_1 
    \right\}.
    \]
\begin{lemma}[Continuity]
    Suppose $\Theta$ is compact and $\int \Vert X'X \Vert q^*(w) dw < \infty$, where $\Vert \cdot \Vert$ denotes the largest absolute value of the eigenvalue. Then the map $\theta \mapsto p_{\theta}$ is continuous with respect to the total variation distance if we define the distance between parameters by $d(\theta_1, \theta_2)$.
\end{lemma}

\begin{proof}[Proof of Lemma 1]
    Let $\theta_1 \ne \theta_2$. There is  $(s,g)\in \left( \{1\} \cup \mathcal{S}\right)\times \mathcal{G}$ with 
    $\psi_{1sg}\ne \psi_{2sg}$ or $\Lambda_{s,1}\ne \Lambda_{s,2}$. Note that 
    \[
    \log p_\theta(y \vert w,S=s,G=g) 
    = 
    -\frac{T}{2}\log 2\pi 
    -\frac{1}{2}\log \det(\Lambda_{s})
    -\frac{1}{2}(y - X \psi_{sg})' \Lambda_{s}^{-1} (y - X \psi_{sg}).
    \]
    If $\Lambda_{s,1}\ne \Lambda_{s,2}$,  clearly $p_{\theta_1}(y \vert w,s,g) \ne p_{\theta_2}(y \vert w,s,g)$ with positive probability. Hence suppose that 
     $\psi_{sg,1}\ne \psi_{sg,2}$ and $\Lambda_{s,1} = \Lambda_{s,2} = \Lambda_{s}$. 
     Note that for $\varepsilon_s \sim N_T(0, \Lambda_{s})$,
     \begin{align*}
        &(y - X \psi_{sg})' \Lambda_{s}^{-1} (y - X \psi_{sg})
        =
        \left\{ X(\psi_{sg} - \psi_{sg}^* ) + \varepsilon_s \right\}'
        \Lambda_{s}^{-1} 
        \left\{ X(\psi_{sg} - \psi_{sg}^* ) + \varepsilon_s \right\}  \\
        &=
        (\psi_{sg} - \psi_{sg}^* )' X'  \Lambda_{s}^{-1} X  (\psi_{sg} - \psi_{sg}^* )
        +
         (\psi_{sg} - \psi_{sg}^* )' X'  \Lambda_{s}^{-1} \varepsilon_s 
        +  
        \varepsilon_s' \Lambda_{s}^{-1} X  (\psi_{sg} - \psi_{sg}^* )
        +
        \varepsilon_s' \Lambda_{s}^{-1} \varepsilon_s 
     \end{align*}
     Hence 
     \[
        \mathbb{E}(y - X \psi_{sg})' \Lambda_{s}^{-1} (y - X \psi_{sg})
         =
         (\psi_{sg} - \psi_{sg}^* )' \mathbb{E}\left( X'  \Lambda_{s}^{-1} X \right) (\psi_{sg} - \psi_{sg}^* )
         +
         \mathbb{E}\left(  \varepsilon_s' \Lambda_{s}^{-1} \varepsilon_s  \right)
     \]    
     Then 
     \begin{align*}
     &\mathbb{E}(y - X \psi_{sg,1})' \Lambda_{s}^{-1} (y - X \psi_{sg,1})
     -
     \mathbb{E}(y - X \psi_{sg,2})' \Lambda_{s}^{-1} (y - X \psi_{sg,2})  \\
     &=
    (\psi_{sg,1} - \psi_{sg,2} )' 
    \mathbb{E}\left( X'  \Lambda_{s}^{-1} X \right) 
    (\psi_{sg,1} - \psi_{sg,2} )   ,           
     \end{align*}     
     which is positive if $\mathbb{E}(X' \Lambda^{-1}_s X)$ is positive definite.
\end{proof}

\begin{proof}[Proof of Lemma 2]
    The total variation distance between $p_{\theta_1}$ and $p_{\theta_2}$ is 
    \[
    d_{TV}(p_{\theta_1},p_{\theta_2})= \int \sum_{g\in\mathcal{G}}\sum_{s\in \{1\} \cup \mathcal{S}}\vert p_{\theta_1}(y \vert w, s,g) - p_{\theta_2}(y \vert w, s,g) \vert h_s^* r_g^* q^*(w) dw dy.
    \]
    By Pinsker’s inequality, the total variation distance is bounded by 2 times the square root of the KL distance. For each cohort $s$ and stratum $g$, the KL distance between two normal distributions 
    $N(y_i \vert X_i \psi_{sg,1}, \Lambda_{s,1})$
    and 
    $N(y_i \vert X_i \psi_{sg,2}, \Lambda_{s,2})$ is, 
    \begin{equation}
    \frac{1}{2}
    \left( 
    \text{tr}\left(\Lambda_{s,2}^{-1}\Lambda_{s,1} - I \right) 
    +
    (\mu_{sg,2} - \mu_{sg,1})' \Lambda_{s,2}^{-1} (\mu_{sg,2} - \mu_{sg,1})
    +
    \log \frac{\det (\Lambda_{s,2})}{\det (\Lambda_{s,1})}
    \right),    \nonumber  
    \end{equation}    
    where $\mu_{sg,k} = X  \psi_{sg,k}$, for $k=1,2$. 
    Denote by $\sigma^2_{s,k} =(\sigma^2_{s1,k},\ldots,\sigma^2_{sT,k})$ , for $k=1,2$. 
    By Lemmas 3-5, the KL distance is bounded by, for some constants $d_1,d_2,d_3>0$,   
    \[
        d_1 \Vert \sigma^2_{s,1} - \sigma^2_{s,2} \Vert_1
        + 
        d_2 \vert D_{s,1} - D_{s,2} \vert   
        +
        d_3 \Vert X' X \Vert \cdot \Vert \psi_{sg,1} - \psi_{sg,2} \Vert_2^2.
    \]
    Therefore, $d_{TV}(p_{\theta_1},p_{\theta_2})$ is bounded by 
    \begin{align*}
    &
    \tilde{d}_1 
    \sum_{sg} 
    h^*_s
    r^*_g
    \left( 
    \Vert \sigma^2_{s,1} - \sigma^2_{s,2} \Vert_1 
    + 
    \vert D_{s,1} - D_{s,2} \vert
    \right) \\
    &+
    \tilde{d}_2 \int \Vert X' X \Vert q^*(w) d w 
    \cdot 
    \sum_{sg} 
    h^*_s
    r^*_g
    \left(
    \Vert \phi_{s,1} - \phi_{s,2} \Vert_2^2 
    +
    \Vert \beta_{sg,1} - \beta_{sg,2} \Vert_2^2 
    \right)
    \\
    &\leq 
    \tilde{d}_3 \Vert \sigma^2_1 - \sigma^2_2 \Vert_1 
    +
    \tilde{d}_4 \Vert D_1 - D_2 \Vert_1
    +
    \tilde{d}_5 \int \Vert X' X \Vert q^*(w) d w 
    \cdot   
    \left(
    \Vert \phi_1 - \phi_2 \Vert_2^2
    +
    \Vert \beta_1 - \beta_2 \Vert_2^2    
    \right)
    \\
    &\leq  
    \tilde{d}_6
    \max 
    \left\{
    \Vert \beta_1 - \beta_2 \Vert_2^2 ,   
    \Vert \phi_1 - \phi_2 \Vert_2^2,    
    \Vert \sigma^2_1 - \sigma^2_2 \Vert_1,
    \Vert D_1 - D_2 \Vert_1    
    \right\}
    \end{align*}    
    for some constants $\tilde{d}_1,\tilde{d}_2,\tilde{d}_3,\tilde{d}_4,\tilde{d}_5,\tilde{d}_6>0$. 
\end{proof}

In Lemmas 3-5 below, we derive  an upper bound for each of the following terms (ignoring the $(s,g)$ subscripts for the ease of notation):
    \begin{equation}
    \frac{1}{2}
    \left( 
    \text{tr}\left(\Lambda_{2}^{-1}\Lambda_{1} - I \right) 
    +
    (\mu_2 - \mu_1)' \Lambda_{2}^{-1} (\mu_2 - \mu_1)
    +
    \log \frac{\det (\Lambda_{2})}{\det (\Lambda_{1})}
    \right),    \label{eq:KL_normal}    
    \end{equation} 
where $\mu_{k} = X \psi_{k}$, for $k=1,2$.    
Recall that $\Lambda = \Sigma + D \cdot 1_T 1_T'$, 
where $\Sigma=\text{diag}(\sigma^2_1,\ldots,\sigma^2_T)$. 
By the Sherman-Morrison formula, we can write
    \begin{equation} 
    \Lambda^{-1} 
    = \Sigma^{-1} - D \frac{\Sigma^{-1} 1_T1_T' \Sigma^{-1}}{1+D \sum_{t=1}^T \frac{1}{\sigma^2_t}}
    \equiv \Sigma^{-1} - M, \label{eq:sherman_morrison}
    \end{equation}
    where $M$ is positive semi-definite. 
Lemma 3 derives an bound for the second term in \eqref{eq:KL_normal}.    
\begin{lemma} 
Suppose $0< \ell_\sigma < \sigma^2_t, \forall t=1,\ldots,T$,  
and 
$\int \Vert X'X \Vert q^*(w) dw < \infty$, where $\Vert \cdot \Vert$ denotes the largest absolute value of the eigenvalue. 
Then, in equation \eqref{eq:KL_normal},
    \begin{align*}
        (\mu_2 - \mu_1)' \Lambda_{2}^{-1} (\mu_2 - \mu_1)
        \leq 
        b\Vert X' X \Vert \cdot \Vert \psi_1 - \psi_2 \Vert_2^2.
    \end{align*}

\end{lemma}
\begin{proof}
    By \eqref{eq:sherman_morrison}, we write $\Lambda^{-1}_2 = \Sigma_2^{-1} - M_2$. Since $M_2$ is positive semi-definite,  $(\mu_2 - \mu_1)' M_2 (\mu_2 - \mu_1) \geq 0$,
    and, for some constant $b>0$, we have 
    \[
    (\mu_2 - \mu_1)' \Lambda_2^{-1} (\mu_2 - \mu_1)
     \leq
     (\mu_2 - \mu_1)' \Sigma_2^{-1} (\mu_2 - \mu_1)
     \leq 
     \frac{\Vert \mu_1 - \mu_2 \Vert_2^2}{\min_t \sigma^2_{2t}}
     \leq 
     b\Vert X' X \Vert \cdot \Vert \psi_1 - \psi_2 \Vert_2^2.
    \]
\end{proof}    
To bound the last term in \eqref{eq:KL_normal}, note that 
by the Matrix Determinant lemma, we have
\begin{equation}
\det (\Lambda) 
= \det(\Sigma) \cdot (1 + D 1_T' \Sigma^{-1} 1_T)
= (\prod_t \sigma^2_{t} ) \cdot (1+D \sum_{t} 1/\sigma^2_t). \label{eq:matrix_determinant_lemma}
\end{equation}
\begin{lemma} 
Suppose 
$0< \ell_\sigma < \sigma^2_t, \forall t=1,\ldots,T$, 
and 
$0 < \ell_D < D$. 
In equation \eqref{eq:KL_normal},
    \begin{align*}
        \log \frac{\det (\Lambda_2) }{\det (\Lambda_1) } 
        \leq  
        c_1 \Vert \sigma^2_1 - \sigma^2_2 \Vert_1 + c_2 \vert D_1 - D_2 \vert,        
    \end{align*}
    for some constants $c_1,c_2>0$.
\end{lemma}
\begin{proof}
    By \eqref{eq:matrix_determinant_lemma}, 
    \begin{align*}
        \log \frac{\det (\Lambda_2) }{\det (\Lambda_1) }
        &=
        \log \left(
        \prod_t 
        \frac{\sigma^2_{2t}}{\sigma^2_{1t}}
        \right)
        \cdot 
        \left(
        \frac{1+D_2 \sum_{t} 1/\sigma^2_{2t}}{1+D_1 \sum_{t} 1/\sigma^2_{1t}}
        \right)   
        =
        \sum_t \log \left( \frac{\sigma^2_{2t}}{\sigma^2_{1t}} \right)
        +
        \log 
        \left(
        \frac{1+D_2 \sum_{t} 1/\sigma^2_{2t}}{1+D_1 \sum_{t} 1/\sigma^2_{1t}}
        \right)         
    \end{align*}
Since $\log x \leq x-1$, the first term is bounded by 
\[
\sum_t \bigg\vert \frac{\sigma^2_{2t}}{\sigma^2_{1t}} -1 \bigg\vert 
=
\sum_t \frac{\vert \sigma^2_{2t} - \sigma^2_{1t}\vert }{\sigma^2_{1t}}
\leq 
\frac{1}{\ell_\sigma} \sum_t \vert \sigma^2_{2t} - \sigma^2_{1t} \vert .
\]    
%where $\ell_\sigma$ is such that $\ell_\sigma<\sigma^2_t, \forall t$. 
Since 
%$\log\frac{1+x}{1+y} \leq |x-y|$,
$\vert \log \frac{1+x}{1+y} \vert \leq \frac{\vert x-y\vert}{1+\min(x,y)}\leq \vert x -y\vert$ for $x,y\geq 0$, 
the second term above is bounded by 
\begin{align*}
    &\vert D_2 \sum_{t} 1/\sigma^2_{2t} - D_1 \sum_{t} 1/\sigma^2_{1t} \pm D_2 \sum_{t} 1/\sigma^2_{1t} \vert
    \leq 
    D_2 \sum_t \vert 1/\sigma^2_{2t}-1/\sigma^2_{1t}\vert + \vert D_2 - D_1 \vert \sum_t 1/\sigma^2_{1t}\\
    & \leq u_D \sum_t \bigg\vert \frac{\sigma^2_{1t}-\sigma^2_{2t}}{\sigma^2_{1t}\sigma^2_{2t}}\bigg\vert 
    +
    \vert D_2 - D_1 \vert T/ \ell_\sigma
    \leq 
    \frac{u_D}{\ell^2_\sigma} \sum_t \vert \sigma^2_{1t}-\sigma^2_{2t} \vert 
    +
    \vert D_2 - D_1 \vert T/ \ell_\sigma  .  
\end{align*}
%where $D<u_D$. 
Hence, 
\[
\log \frac{\det (\Lambda_2) }{\det (\Lambda_1) }
\leq 
c_1 \sum_t \vert \sigma^2_{1t} - \sigma^2_{2t} \vert 
+
c_2   \vert D_1 - D_2 \vert .
\]

\end{proof}
Lemma 5 derives a bound for the first term in \eqref{eq:KL_normal}.    
\begin{lemma}  
Suppose $0< \ell_\sigma < \sigma^2_t < u_\sigma < \infty, \forall t=1,\ldots,T$, 
and 
$0< \ell_D < D < u_D < \infty$. 
In equation \eqref{eq:KL_normal},
\begin{align*}
        \vert \text{tr}\left(\Lambda_{2}^{-1}\Lambda_{1} - I \right) \vert
        \leq 
        a_1 \Vert \sigma^2_1 - \sigma^2_2 \Vert_1 + a_2 \vert D_1 - D_2 \vert,
    \end{align*}
    for some constants $a_1,a_2>0$.
\end{lemma}
\begin{proof}
Note that 
\[
\Lambda_{2}^{-1}\Lambda_{1} - I 
=
\Lambda_{2}^{-1}(\Lambda_{1}-\Lambda_{2})
=
\Lambda_{2}^{-1}\left( (\Sigma_1 - \Sigma_2) + (D_1-D_2) 1_T 1_T'\right)
\]
Hence 
\[
\text{tr}\left(\Lambda_{2}^{-1}\Lambda_{1} - I \right) 
=
\text{tr}\left(\Lambda_{2}^{-1} (\Sigma_1 - \Sigma_2) \right) 
+
(D_1-D_2) \text{tr} (\Lambda_{2}^{-1}1_T 1_T' ).
\]
The first term by the Cauchy-Schwarz theorem is bounded by 
\[
\vert \vert \Lambda_{2}^{-1} \vert \vert_F 
\cdot 
\vert \vert \Sigma_1 - \Sigma_2 \vert \vert_F
\leq 
c \cdot \sqrt{\sum_t \vert \sigma^2_{1t}- \sigma^2_{2t} \vert^2}
=
c \cdot \Vert \sigma^2_1 - \sigma^2_2 \Vert_2
\leq 
c \cdot \Vert \sigma^2_1 - \sigma^2_2 \Vert_1,
\]
where $\vert \vert \cdot \vert \vert_F$ is the Frobenius norm.
By the property of trace, the second term is bounded by $\vert D_1-D_2 \vert$ times 
\begin{align*}
    \text{tr} (1_T' \Lambda_{2}^{-1}1_T ) 
   &=
    1_T' \Lambda_{2}^{-1}1_T  
    =
    1_T' (\Sigma_2^{-1} - M_2)1_T  
    \leq 
    1_T' \Sigma_2^{-1}1_T    
    = 
    \sum 1/\sigma^2_{2t}
    \leq
    T/\ell_\sigma
\end{align*}
where the positive semi-definite matrix $M_2$ is defined in \eqref{eq:sherman_morrison}.
\end{proof}

\clearpage
\section{Additional simulations}\label{sec:additional_simulation}
\subsection{More on simulation 1}
\noindent The table below shows simulation results under a large $n$.
\begin{table}[!htbp]
\centering
\caption{Simulation comparison: $G=3$, prePT=0, $n=250$, $R=300$}
\label{tab:sim_comparison_G3_prePT0_N250_R300}
\begin{threeparttable}
\resizebox{\textwidth}{!}{%
\begin{tabular}{lrrrrrrrrrr}
\hline\hline
ATT & True & BDID Bias & BDID RMSE & BDID Cov. & CS Split Bias & CS Split RMSE & CS Split Cov. & CS Pooled Bias & CS Pooled RMSE & CS Pooled Cov. \\
\hline
\multicolumn{11}{l}{\textit{Panel A: Individual ATT comparison}} \\
\hline
$\tau_{\mathrm{ATT}}(2,2; 1)$& -0.0209 & -0.0058 & 0.0778 & 0.967 & -0.0016 & 0.0580 & 0.930 & -0.0536 & 0.0651 & 0.713 \\
$\tau_{\mathrm{ATT}}(2,3; 1)$& -0.0810 & -0.0008 & 0.0667 & 0.970 & -0.0010 & 0.0528 & 0.940 & -0.2119 & 0.2197 & 0.047 \\
$\tau_{\mathrm{ATT}}(2,4; 1)$& -0.1434 & 0.0025 & 0.0757 & 0.973 & 0.0005 & 0.0596 & 0.960 & -0.3779 & 0.3885 & 0.013 \\
$\tau_{\mathrm{ATT}}(2,5; 1)$& -0.1074 & -0.0055 & 0.0812 & 0.960 & -0.0044 & 0.0652 & 0.940 & -0.2854 & 0.2940 & 0.027 \\
$\tau_{\mathrm{ATT}}(4,4; 1)$& 0.0035 & 0.0038 & 0.0630 & 0.983 & 0.0013 & 0.0491 & 0.943 & 0.0110 & 0.0317 & 0.917 \\
$\tau_{\mathrm{ATT}}(4,5; 1)$& -0.0350 & 0.0009 & 0.0623 & 0.983 & -0.0023 & 0.0504 & 0.960 & -0.0952 & 0.1013 & 0.220 \\
$\tau_{\mathrm{ATT}}(5,5; 1)$& -0.0271 & -0.0038 & 0.0926 & 0.973 & -0.0093 & 0.0768 & 0.923 & -0.0761 & 0.0878 & 0.637 \\
$\tau_{\mathrm{ATT}}(2,2; 2)$& -0.2091 & 0.0051 & 0.0766 & 0.973 & 0.0044 & 0.0623 & 0.917 & 0.1346 & 0.1395 & 0.060 \\
$\tau_{\mathrm{ATT}}(2,3; 2)$& -0.8104 & 0.0056 & 0.0748 & 0.953 & 0.0038 & 0.0571 & 0.953 & 0.5174 & 0.5206 & 0.000 \\
$\tau_{\mathrm{ATT}}(2,4; 2)$& -1.4337 & 0.0034 & 0.0784 & 0.970 & 0.0011 & 0.0636 & 0.930 & 0.9124 & 0.9168 & 0.000 \\
$\tau_{\mathrm{ATT}}(2,5; 2)$& -1.0738 & 0.0085 & 0.0830 & 0.957 & 0.0018 & 0.0643 & 0.943 & 0.6810 & 0.6847 & 0.000 \\
$\tau_{\mathrm{ATT}}(4,4; 2)$& 0.0349 & 0.0024 & 0.0615 & 0.987 & 0.0011 & 0.0518 & 0.920 & -0.0204 & 0.0360 & 0.880 \\
$\tau_{\mathrm{ATT}}(4,5; 2)$& -0.3503 & 0.0086 & 0.0649 & 0.977 & 0.0021 & 0.0495 & 0.957 & 0.2200 & 0.2227 & 0.000 \\
$\tau_{\mathrm{ATT}}(5,5; 2)$& -0.2713 & -0.0008 & 0.0915 & 0.960 & -0.0016 & 0.0727 & 0.940 & 0.1681 & 0.1737 & 0.043 \\
$\tau_{\mathrm{ATT}}(2,2; 3)$ & -0.0021 & 0.0044 & 0.0742 & 0.973 & 0.0018 & 0.0592 & 0.930 & -0.0724 & 0.0813 & 0.553 \\
$\tau_{\mathrm{ATT}}(2,3; 3)$ & -0.0081 & 0.0054 & 0.0658 & 0.970 & 0.0040 & 0.0523 & 0.927 & -0.2849 & 0.2907 & 0.000 \\
$\tau_{\mathrm{ATT}}(2,4; 3)$ & -0.0143 & 0.0031 & 0.0773 & 0.970 & 0.0009 & 0.0608 & 0.947 & -0.5070 & 0.5149 & 0.000 \\
$\tau_{\mathrm{ATT}}(2,5; 3)$ & -0.0107 & -0.0007 & 0.0841 & 0.937 & -0.0017 & 0.0652 & 0.937 & -0.3820 & 0.3885 & 0.000 \\
$\tau_{\mathrm{ATT}}(4,4; 3)$ & 0.0003 & -0.0036 & 0.0668 & 0.960 & -0.0003 & 0.0544 & 0.913 & 0.0142 & 0.0329 & 0.897 \\
$\tau_{\mathrm{ATT}}(4,5; 3)$ & -0.0035 & -0.0004 & 0.0658 & 0.963 & -0.0003 & 0.0525 & 0.933 & -0.1267 & 0.1313 & 0.050 \\
$\tau_{\mathrm{ATT}}(5,5; 3)$ & -0.0027 & 0.0045 & 0.0903 & 0.980 & 0.0030 & 0.0758 & 0.913 & -0.1005 & 0.1097 & 0.403 \\
\hline
\multicolumn{11}{l}{\textit{Panel B: Overall comparison}} \\
\hline
Method & -- & Mean Bias & Max $|$Bias$|$ & RMSE & SD & Coverage & -- & -- & -- & -- \\
BDID & -- & 0.0017 & 0.0086 & 0.0750 & 0.0836 & 0.969 & -- & -- & -- & -- \\
CS split & -- & 0.0002 & 0.0093 & 0.0597 & 0.0597 & 0.936 & -- & -- & -- & -- \\
CS pooled & -- & 0.0031 & 0.9124 & 0.2586 & 0.0520 & 0.260 & -- & -- & -- & -- \\
\hline\hline
\end{tabular}}
\begin{tablenotes}[flushleft]
\footnotesize
\item \textit{Notes:} BDID denotes the proposed estimator. CS Split and CS Pooled denote comparison estimators. 
\item Bias, RMSE, and coverage are computed over 300 simulation replications. 
\item Coverage refers to nominal 95\% credible/confidence interval coverage.
\end{tablenotes}
\end{threeparttable}
\end{table}
\FloatBarrier
\subsection{More on simulation 2}
\noindent The Table below shows simulation results when prePT = 1.
\FloatBarrier
\begin{table}[htbp]
\centering
\caption{ATT estimation results across sample sizes: $G=2$, prePT$_{DGP}=1$, prePT$_{EST}=1$}
\label{tab:ATT_combined_G2_prePT1}
\scriptsize
\begin{threeparttable}
\resizebox{\textwidth}{!}{%
\begin{tabular}{lrrrrrrrrrrrr}
\toprule
 & \multicolumn{3}{c}{Bias} & \multicolumn{3}{c}{RMSE} & \multicolumn{3}{c}{Post. SD} & \multicolumn{3}{c}{Coverage} \\
\cmidrule(lr){2-4} \cmidrule(lr){5-7} \cmidrule(lr){8-10} \cmidrule(lr){11-13}
ATT & $n=250$ & $n=500$ & $n=1000$ & $n=250$ & $n=500$ & $n=1000$ & $n=250$ & $n=500$ & $n=1000$ & $n=250$ & $n=500$ & $n=1000$ \\
\midrule
$\tau_{\mathrm{ATT}}(2,2; 1)$& -0.002 & 0.001 & -0.000 & 0.046 & 0.033 & 0.023 & 0.049 & 0.034 & 0.023 & 0.955 & 0.953 & 0.954 \\
$\tau_{\mathrm{ATT}}(2,3; 1)$& -0.001 & 0.001 & -0.001 & 0.042 & 0.030 & 0.021 & 0.046 & 0.031 & 0.022 & 0.967 & 0.957 & 0.958 \\
$\tau_{\mathrm{ATT}}(2,4; 1)$& -0.003 & 0.002 & -0.000 & 0.048 & 0.034 & 0.024 & 0.051 & 0.035 & 0.024 & 0.958 & 0.953 & 0.960 \\
$\tau_{\mathrm{ATT}}(2,5; 1)$& -0.002 & 0.000 & -0.000 & 0.055 & 0.038 & 0.026 & 0.055 & 0.038 & 0.027 & 0.941 & 0.948 & 0.953 \\
$\tau_{\mathrm{ATT}}(4,4; 1)$& 0.001 & 0.001 & 0.001 & 0.032 & 0.023 & 0.016 & 0.036 & 0.024 & 0.017 & 0.968 & 0.971 & 0.971 \\
$\tau_{\mathrm{ATT}}(4,5; 1)$& 0.001 & -0.001 & 0.000 & 0.039 & 0.028 & 0.020 & 0.042 & 0.029 & 0.020 & 0.972 & 0.957 & 0.949 \\
$\tau_{\mathrm{ATT}}(5,5; 1)$& 0.001 & -0.000 & 0.001 & 0.053 & 0.038 & 0.026 & 0.053 & 0.037 & 0.026 & 0.951 & 0.944 & 0.955 \\
$\tau_{\mathrm{ATT}}(2,2; 2)$& -0.001 & 0.001 & 0.000 & 0.046 & 0.033 & 0.024 & 0.049 & 0.034 & 0.023 & 0.959 & 0.953 & 0.948 \\
$\tau_{\mathrm{ATT}}(2,3; 2)$& 0.000 & 0.000 & 0.000 & 0.041 & 0.029 & 0.021 & 0.046 & 0.031 & 0.022 & 0.974 & 0.967 & 0.957 \\
$\tau_{\mathrm{ATT}}(2,4; 2)$& 0.000 & 0.001 & 0.000 & 0.047 & 0.033 & 0.024 & 0.051 & 0.035 & 0.024 & 0.967 & 0.963 & 0.956 \\
$\tau_{\mathrm{ATT}}(2,5; 2)$& 0.004 & 0.002 & 0.001 & 0.053 & 0.037 & 0.026 & 0.056 & 0.038 & 0.027 & 0.956 & 0.948 & 0.943 \\
$\tau_{\mathrm{ATT}}(4,4; 2)$& 0.001 & 0.000 & 0.001 & 0.033 & 0.023 & 0.016 & 0.037 & 0.024 & 0.017 & 0.968 & 0.960 & 0.960 \\
$\tau_{\mathrm{ATT}}(4,5; 2)$& 0.005 & 0.001 & 0.001 & 0.040 & 0.029 & 0.019 & 0.043 & 0.029 & 0.020 & 0.961 & 0.940 & 0.964 \\
$\tau_{\mathrm{ATT}}(5,5; 2)$& 0.004 & 0.003 & 0.001 & 0.052 & 0.037 & 0.025 & 0.053 & 0.037 & 0.026 & 0.953 & 0.942 & 0.959 \\
\midrule
OVERALL & 0.000 & 0.001 & 0.000 & 0.045 & 0.032 & 0.022 & 0.048 & 0.033 & 0.023 & 0.961 & 0.954 & 0.956 \\
\bottomrule
\end{tabular}
}
\begin{tablenotes}[flushleft]
\footnotesize
\item Notes: The table reports bias, RMSE, average posterior standard deviation, 
\item and empirical coverage of 95\% credible intervals across 1000 Monte Carlo repetitions. 
\end{tablenotes}
\end{threeparttable}
\end{table}
\FloatBarrier

\subsection{Simulation 3: marginal likelihood and specification selection}
We 
proposed marginal likelihood comparison as a probabilistic framework for evaluating alternative identifying assumptions and heterogeneity specifications. In this simulation study, we examine the finite-sample performance of this approach in two settings. Part A studies whether the marginal likelihood can correctly distinguish between models with and without pre-treatment parallel trends (\texttt{prePT}) restrictions when the degree of heterogeneity is fixed. This corresponds to settings in which the researcher already has a substantive stratification in mind. Part B then studies the more general problem of jointly selecting the number of strata $G$ and the pre-treatment specification. In particular, we examine whether the marginal likelihood appropriately penalizes unnecessarily rich heterogeneity structures while still detecting meaningful subgroup heterogeneity when supported by the data.
\paragraph{Part A: prePT selection conditional on correct $G$}
In this section, we examine the ability of the marginal likelihood to distinguish between whether parallel trends hold in the pre-treatment periods (\texttt{prePT}=1) or not (\texttt{prePT}=0). For each simulation design, we estimate the two competing models (\texttt{prePT}=0 vs 1) and compare them using  log marginal likelihood. The structure of stratification used in estimation is fixed at the truth in the data generating process.

Table \ref{tab:ML_selection_allG_N200} shows that the marginal likelihood consistently selects the correct pre-treatment specification across all values of $G$. When parallel trends do not hold in pre-periods in the true DGP, (\texttt{prePT}=0), the model imposing \texttt{prePT}=0 is selected in almost all Monte Carlo repetitions, with win shares between 0.93 and 1.00. Conversely, when parallel trends hold in the pre-treatment periods (\texttt{prePT}=1), the corresponding \texttt{prePT}=1 specification is selected in all repetitions for every value of G. The differences in average log marginal likelihood are substantial, indicating strong evidence in favor of the correctly specified model. Overall, the results suggest that the marginal likelihood provides reliable model selection performance for detecting whether pre-treatment parallel trends restrictions are supported by the data, even in relatively small samples ($n=200$).
\FloatBarrier
\begin{table}[htbp]
\centering
\caption{Model selection results for $n=200$}
\label{tab:ML_selection_allG_N200}
\scriptsize
\begin{threeparttable}
%\begin{center}
\resizebox{\textwidth}{!}{%
\begin{tabular}{llrrrrr}
\toprule
DGP $G$ & DGP prePT & Model & Avg. log ML & SD log ML & Wins & Win share \\
\midrule
1 & 0 & G1\_prePT0 & -60.07 & 22.02 & 48 & 0.96 \\
1 & 0 & G1\_prePT1 & -78.87 & 23.33 & 2 & 0.04 \\
\midrule
1 & 1 & G1\_prePT0 & -40.69 & 22.87 & 1 & 0.02 \\
1 & 1 & G1\_prePT1 & -31.55 & 22.12 & 49 & 0.98 \\
\midrule
2 & 0 & G2\_prePT0 & -88.19 & 22.32 & 50 & 1.00 \\
2 & 0 & G2\_prePT1 & -123.62 & 27.45 & 0 & 0.00 \\
\midrule
2 & 1 & G2\_prePT0 & -67.33 & 23.03 & 0 & 0.00 \\
2 & 1 & G2\_prePT1 & -49.20 & 20.97 & 50 & 1.00 \\
\midrule
3 & 0 & G3\_prePT0 & -111.01 & 21.09 & 50 & 1.00 \\
3 & 0 & G3\_prePT1 & -134.69 & 22.79 & 0 & 0.00 \\
\midrule
3 & 1 & G3\_prePT0 & -90.47 & 22.32 & 0 & 0.00 \\
3 & 1 & G3\_prePT1 & -66.43 & 20.18 & 50 & 1.00 \\
\bottomrule
\end{tabular}
}
%\end{center}
\begin{tablenotes}[flushleft]
\footnotesize
\item Notes: 
%For each value of $G$, estimation fixes the number of strata at the true DGP value, so $G_{EST}=G_{DGP}$. 
For each  specification, the table compares models with prePT equal to 0 and 1 using the marginal likelihood. \item Wins report the number of Monte Carlo repetitions in which the model achieved the highest log marginal likelihood.
\end{tablenotes}
\end{threeparttable}
\end{table}

\FloatBarrier

\paragraph{Part B: Joint selection of $G$ and prePT}
The goal of this simulation study is to study whether marginal likelihood comparison can distinguish between alternative heterogeneity structures and pre-treatment restrictions. In particular, we examine whether the proposed framework appropriately penalizes unnecessarily fine stratifications while still detecting substantively relevant heterogeneity when supported by the data. 

Table \ref{tab:ML_selection_N250} shows that 
the marginal likelihood successfully distinguishes between alternative heterogeneity structures and pre-treatment specifications. Across all simulation designs, the model corresponding to the true ($G$, \texttt{prePT}) combination consistently attains the highest average log marginal likelihood and is selected in nearly all Monte Carlo repetitions. Importantly, the results indicate that the marginal likelihood appropriately penalizes unnecessarily fine stratifications. For example, when the true data-generating process contains only one stratum ($G=1$), richer specifications with $G=2$ or $G=3$ receive substantially lower marginal likelihood values despite their greater flexibility. Conversely, when treatment effect heterogeneity is genuinely present ($G=2$ or $3$), the marginal likelihood favors the corresponding richer specification once sufficient information is available in the data. These findings suggest that the proposed framework provides a useful probabilistic criterion for assessing the degree of subgroup heterogeneity supported by the data while balancing model flexibility against statistical precision.

\FloatBarrier
\begin{table}[htbp!]
\centering
\caption{Model selection results for $n=250$}
\label{tab:ML_selection_N250}
\scriptsize
\begin{threeparttable}
\resizebox{\textwidth}{!}{%
\begin{tabular}{lrrrrrr}
\toprule
DGP $G$ & DGP prePT & Model & Avg. log ML & SD log ML & Wins & Win share \\
\midrule
1 & 0 & G1\_prePT0 & -49.11 & 26.16 & 50 & 1.00 \\
1 & 0 & G1\_prePT1 & -78.18 & 27.38 & 0 & 0.00 \\
1 & 0 & G2\_prePT0 & -80.55 & 24.50 & 0 & 0.00 \\
1 & 0 & G2\_prePT1 & -99.97 & 23.08 & 0 & 0.00 \\
1 & 0 & G3\_prePT0 & -110.77 & 28.79 & 0 & 0.00 \\
1 & 0 & G3\_prePT1 & -124.91 & 27.88 & 0 & 0.00 \\
\midrule
1 & 1 & G1\_prePT0 & -24.36 & 27.53 & 0 & 0.00 \\
1 & 1 & G1\_prePT1 & -12.80 & 28.26 & 50 & 1.00 \\
1 & 1 & G2\_prePT0 & -55.05 & 25.70 & 0 & 0.00 \\
1 & 1 & G2\_prePT1 & -36.34 & 23.56 & 0 & 0.00 \\
1 & 1 & G3\_prePT0 & -85.82 & 28.40 & 0 & 0.00 \\
1 & 1 & G3\_prePT1 & -59.83 & 30.43 & 0 & 0.00 \\
\midrule
2 & 0 & G1\_prePT0 & -228.37 & 21.76 & 0 & 0.00 \\
2 & 0 & G1\_prePT1 & -270.73 & 21.39 & 0 & 0.00 \\
2 & 0 & G2\_prePT0 & -82.36 & 24.87 & 50 & 1.00 \\
2 & 0 & G2\_prePT1 & -133.96 & 25.04 & 0 & 0.00 \\
2 & 0 & G3\_prePT0 & -196.34 & 26.61 & 0 & 0.00 \\
2 & 0 & G3\_prePT1 & -236.37 & 23.44 & 0 & 0.00 \\
\midrule
2 & 1 & G1\_prePT0 & -193.82 & 22.59 & 0 & 0.00 \\
2 & 1 & G1\_prePT1 & -182.44 & 22.33 & 0 & 0.00 \\
2 & 1 & G2\_prePT0 & -54.99 & 24.97 & 0 & 0.00 \\
2 & 1 & G2\_prePT1 & -34.99 & 24.89 & 50 & 1.00 \\
2 & 1 & G3\_prePT0 & -167.20 & 25.38 & 0 & 0.00 \\
2 & 1 & G3\_prePT1 & -145.37 & 22.29 & 0 & 0.00 \\
\midrule
3 & 0 & G1\_prePT0 & -226.09 & 23.46 & 0 & 0.00 \\
3 & 0 & G1\_prePT1 & -255.81 & 24.00 & 0 & 0.00 \\
3 & 0 & G2\_prePT0 & -255.41 & 24.84 & 0 & 0.00 \\
3 & 0 & G2\_prePT1 & -281.11 & 22.97 & 0 & 0.00 \\
3 & 0 & G3\_prePT0 & -110.97 & 29.27 & 50 & 1.00 \\
3 & 0 & G3\_prePT1 & -146.63 & 30.02 & 0 & 0.00 \\
\midrule
3 & 1 & G1\_prePT0 & -197.27 & 23.68 & 0 & 0.00 \\
3 & 1 & G1\_prePT1 & -184.27 & 22.86 & 0 & 0.00 \\
3 & 1 & G2\_prePT0 & -223.81 & 24.77 & 0 & 0.00 \\
3 & 1 & G2\_prePT1 & -206.22 & 21.81 & 0 & 0.00 \\
3 & 1 & G3\_prePT0 & -84.04 & 29.64 & 0 & 0.00 \\
3 & 1 & G3\_prePT1 & -60.30 & 28.45 & 50 & 1.00 \\
\bottomrule
\end{tabular}
}
\begin{tablenotes}[flushleft]
\footnotesize
\item Notes: Each row reports the average marginal likelihood (log ML), its standard deviation across Monte Carlo repetitions, and the frequency with which the model achieved the highest log ML.
\end{tablenotes}
\end{threeparttable}
\end{table}

\clearpage
\section{Additional results for the application}
\FloatBarrier
\begin{table}[htbp]
\centering
\caption{Callaway--Sant'Anna ATT Estimates: Strata and Full Sample}
\label{tab:cs21_side_by_side}
\scriptsize
\begin{threeparttable}
\resizebox{\textwidth}{!}{%
\begin{tabular}{lrrrrrrrr}
\hline\hline
ATT & \multicolumn{4}{c}{Stratum-specific} & \multicolumn{4}{c}{Full sample} \\
 & Mean & SE & LB & UB & Mean & SE & LB & UB \\
\hline
$\tau_{\mathrm{ATT}}(2,2; 1)$ & -0.018 & 0.044 & -0.104 & 0.068 & -0.011 & 0.025 & -0.059 & 0.038 \\
$\tau_{\mathrm{ATT}}(2,3; 1)$ & -0.128 & 0.047 & -0.220 & -0.035 & -0.070 & 0.034 & -0.137 & -0.004 \\
$\tau_{\mathrm{ATT}}(2,4; 1)$ & -0.205 & 0.063 & -0.328 & -0.082 & -0.137 & 0.039 & -0.213 & -0.061 \\
$\tau_{\mathrm{ATT}}(2,5; 1)$ & -0.133 & 0.064 & -0.259 & -0.007 & -0.101 & 0.035 & -0.170 & -0.031 \\
$\tau_{\mathrm{ATT}}(4,4; 1)$ & -0.026 & 0.029 & -0.084 & 0.032 & -0.005 & 0.017 & -0.039 & 0.029 \\
$\tau_{\mathrm{ATT}}(4,5; 1)$ & -0.061 & 0.043 & -0.146 & 0.024 & -0.041 & 0.021 & -0.083 & 0.000 \\
$\tau_{\mathrm{ATT}}(5,5; 1)$ & -0.009 & 0.036 & -0.080 & 0.063 & -0.026 & 0.018 & -0.060 & 0.008 \\
$\tau_{\mathrm{ATT}}(2,2; 2)$ & -0.006 & 0.019 & -0.043 & 0.030 & -- & -- & -- & -- \\
$\tau_{\mathrm{ATT}}(2,3; 2)$ & -0.017 & 0.024 & -0.065 & 0.031 & -- & -- & -- & -- \\
$\tau_{\mathrm{ATT}}(2,4; 2)$ & -0.073 & 0.027 & -0.125 & -0.020 & -- & -- & -- & -- \\
$\tau_{\mathrm{ATT}}(2,5; 2)$ & -0.074 & 0.021 & -0.116 & -0.032 & -- & -- & -- & -- \\
$\tau_{\mathrm{ATT}}(4,4; 2)$ & 0.010 & 0.025 & -0.039 & 0.060 & -- & -- & -- & -- \\
$\tau_{\mathrm{ATT}}(4,5; 2)$ & -0.034 & 0.022 & -0.078 & 0.009 & -- & -- & -- & -- \\
$\tau_{\mathrm{ATT}}(5,5; 2)$ & -0.044 & 0.014 & -0.071 & -0.017 & -- & -- & -- & -- \\
\hline\hline
\end{tabular}}
\begin{tablenotes}[flushleft]
\footnotesize
\item Notes: Stratum-specific estimates apply Callaway--Sant'Anna separately within each stratum. The full-sample estimates use all counties. Stratum 1 is defined by $lpop < 3.2578$ and Stratum 2 by $lpop \geq 3.2578$. Standard errors are clustered at the county level.
\end{tablenotes}
\end{threeparttable}
\end{table}
\FloatBarrier

%\fi 
\end{document}